\documentclass[12pt, draftclsnofoot, onecolumn]{IEEEtran}

\usepackage{amsmath,epsfig,amssymb,verbatim,amsopn,cite,multirow}
\usepackage{amsthm}
\usepackage{balance,nopageno}
\usepackage[usenames,dvipsnames]{color}
\usepackage[all]{xy}  
\usepackage{url}
\usepackage{amsfonts}
\usepackage{amssymb}
\usepackage{epsfig}
\usepackage{epstopdf}
\usepackage{bm}
\usepackage{balance}
\usepackage{graphicx}
\usepackage{subcaption}
\usepackage{footnote}
\usepackage[norelsize, linesnumbered, ruled, lined, boxed, commentsnumbered]{algorithm2e}
\usepackage[margin=19mm,top=20mm,bottom=20mm]{geometry}

\newtheorem{Lemma}{Lemma}

\newtheorem{Corollary}[Lemma]{Corollary}

\newtheorem{Remark}{Remark}

  {\proof}{\proofend}
\newtheorem{proposition}{Proposition}

\DeclareMathOperator*{\argmax}{arg\,max}
\DeclareMathOperator*{\argmin}{arg\,min}

\newcommand{\qh}{{\bf h}}

\newcommand{\qC}{{\bf C}}

\newcommand{\qH}{{\bf H}}

\newcommand{\Ss}{\mathtt{S1}}
\newcommand{\Sss}{\mathtt{S2}}
\newcommand{\Ssss}{\mathtt{S3}}

\newcommand{\hRU}{\qh_{\mathtt{RU_1}}}
\newcommand{\hRUU}{\qh_{\mathtt{RU_2}}}

\newcommand{\HSIkj}{h_{\mathtt{RR}}^{k,j}}
\newcommand{\GSIkj}{\gamma_{\mathtt{SI}}^{k,j}}
\newcommand{\GSIkkj}{\gamma_{\mathtt{SI}}^{k^*,j}}
\newcommand{\bGSI}{\bar{\gamma}_{\mathtt{SI}}}

\newcommand{\HSI}{\qH_{\mathtt{RR}}}
\newcommand{\hBR}{\qH_{\mathtt{SR}}}
\newcommand{\HBRij}{h_{\mathtt{SR}}^{i,j}}
\newcommand{\HBRisj}{h_{\mathtt{SR}}^{i^*,j}}
\newcommand{\GBRij}{\gamma_{\mathtt{SR}}^{i,j}}
\newcommand{\GBRiij}{\gamma_{\mathtt{SR}}^{i^*,j}}
\newcommand{\bGBR}{\bar{\gamma}_{\mathtt{SR}}}

\newcommand{\HBNu}{h_{\mathtt{SU1}}^{i}}
\newcommand{\GSNu}{\gamma_{\mathtt{SU1}}^{i}}
\newcommand{\bGSNu}{\bar{\gamma}_{\mathtt{SU1}}}

\newcommand{\HRNu}{h_{\mathtt{RU1}}^{k}}
\newcommand{\GRNu}{\gamma_{\mathtt{RU1}}^{k}}
\newcommand{\bGRNu}{\bar{\gamma}_{\mathtt{RU1}}}

\newcommand{\HRFu}{h_{\mathtt{RU2}}^{k}}
\newcommand{\GRFu}{\gamma_{\mathtt{RU2}}^{k}}
\newcommand{\GRFui}{\gamma_{\mathtt{RU2,AS}}}
\newcommand{\GRFus}{\gamma_{\mathtt{RU2,S1}}}
\newcommand{\GRFuss}{\gamma_{\mathtt{RU2,S2}}}

\newcommand{\GMR}{\gamma_{\mathtt{R}}}
\newcommand{\GMRi}{\gamma_{\mathtt{R,AS}}}
\newcommand{\GMRS}{\gamma_{\mathtt{R,S1}}}
\newcommand{\GMRSS}{\gamma_{\mathtt{R,S2}}}
\newcommand{\bGRFu}{\bar{\gamma}_{\mathtt{RU2}}}

\newcommand{\SRFu}{\sigma^2_{\mathtt{RU2}}}
\newcommand{\SRNu}{k_1\sigma^2_{\mathtt{RU1}}}
\newcommand{\SBNu}{\sigma^2_{\mathtt{SU1}}}
\newcommand{\SBR}{{\sigma}^2_{\mathtt{SR}}}

\newcommand{\Pkk}{\mathtt{P}_k}
\newcommand{\Pii}{\mathtt{P}_i}
\newcommand{\Pjj}{\mathtt{P}_j}
\newcommand{\Pll}{\mathtt{P}_\ell}
\newcommand{\Pmm}{\mathtt{P}_m}

\newcommand{\cno}{\Psi(\theta_2)}
\newcommand{\cnox}{\Psi(x)}

\newcommand{\Snuu}{\sigma_{n_{2}}^2}
\newcommand{\Sn}{\sigma_n^2}
\newcommand{\Sap}{\sigma_{\mathtt{SI}}^2}
\newcommand{\SXY}{\sigma_{\mathtt{XY}}^2}

\newcommand{\PS}{P_{\mathtt{S}}}
\newcommand{\PR}{P_{\mathtt{R}}}

\newcommand{\SUu}{\text{U1}}
\newcommand{\SUuu}{\text{U2}}
\newcommand{\SUui}{\text{U}u}

\newcommand{\Rly}{\text{R}}
\newcommand{\Srs}{\mathtt{S}}
\newcommand{\Nu}{\mathtt{U1}}
\newcommand{\Fu}{\mathtt{U2}}
\newcommand{\Poutnu}{{\mathtt{P_{out,1}^{AS}}}}
\newcommand{\Poutfu}{{\mathtt{P_{out,2}^{AS}}}}

\newcommand{\PoutnuS}{{\mathtt{P_{out,1}^{S1}}}}
\newcommand{\PoutfuS}{{\mathtt{P_{out,2}^{S1}}}}
\newcommand{\PoutnuSS}{{\mathtt{P_{out,1}^{S2}}}}
\newcommand{\PoutfuSS}{{\mathtt{P_{out,2}^{S2}}}}
\newcommand{\PoutfuSSS}{{\mathtt{P_{out,2}^{S3}}}}
\newcommand{\PoutnuSSS}{{\mathtt{P_{out,1}^{S3}}}}

\newcommand{\MT}{M_{\mathtt{T}}}
\newcommand{\MR}{M_{\mathtt{R}}}
\newcommand{\MB}{N_{\mathtt{T}}}

\newcommand{\Prob}{\textnormal{Pr}}

\newcommand{\RnuAS}{{R}_{\SUu}^{\mathtt{AS}}}
\newcommand{\RfuAS}{{R}_{\SUuu}^{\mathtt{AS}}}

\newcommand{\be}{\begin{equation}} \newcommand{\ee}{\end{equation}}
\newcommand{\bea}{\begin{eqnarray}} \newcommand{\eea}{\end{eqnarray}}

\usepackage{multibib}
\usepackage[nodisplayskipstretch]{setspace}
\newcites{Prim}{Very important papers}

\definecolor{light-gray}{gray}{0.65}

\newcounter{mytempeqcounter}

\title{\fontsize{0.85cm}{1cm}\selectfont New Antenna Selection Schemes for Full-Duplex Cooperative MIMO-NOMA Systems
}

\author{\normalsize {Zahra Mobini,~\IEEEmembership{Member,~IEEE,}
Mohammadali Mohammadi,~\IEEEmembership{Member,~IEEE,}\\
Theodoros A. Tsiftsis,~\IEEEmembership{Senior Member,~IEEE,}
   Zhiguo Ding,~\IEEEmembership{Fellow,~IEEE}\\
  and Chintha Tellambura,~\IEEEmembership{Fellow,~IEEE.}}
  \thanks{
    Z. Mobini and M. Mohammadi  are with the Faculty of Engineering, Shahrekord University, Shahrekord 115, Iran
    (email: \{m.a.mohammadi, z.mobini\}@sku.ac.ir).}
  \thanks{
        Theodoros A. Tsiftsis is with the School of Intelligent Systems Science and
        Engineering, Jinan University, Zhuhai 519070, China, and also with the Department
        of Computer Science and Telecommunications, University of Thessaly,
        Lamia 35131, Greece (e-mail: theo\_tsiftsis@jnu.edu.cn). }
    \thanks{Z. Ding is with the School of Electrical and Electronic Engineering, the University of Manchester, Manchester, UK (e-mail: zhiguo.ding@manchester.ac.uk).}
    \thanks{ C. Tellambura is with the Department of Electrical and Computer Engineering,
University of Alberta, Edmonton, AB T6G 2V4 Canada  (email:
    {chintha@ece.ualberta.ca).}}
\thanks{%
Part of this work was presented at the IEEE International Conference on Communications (ICC 2018), Kansas City, MO, USA, May. 2018~\cite{Mohammadi:icc:2018}.
}
}

\begin{document}

\maketitle
\thispagestyle{empty}
\vspace{-2em}
\begin{abstract}
In this paper, we address the antenna selection (AS) problem in full-duplex (FD) cooperative non-orthogonal multiple access (NOMA) systems, where a multi-antenna FD relay bridges the connection between the multi-antenna base station and NOMA far user. Specifically, two AS schemes, namely max-$\SUu$ and max-$\SUuu$, are proposed to maximize the end-to-end signal-to-interference-plus-noise ratio at either or both near and far users, respectively. Moreover, a two-stage AS scheme, namely quality-of-service (QoS) provisioning  scheme, is designed to realize a specific rate at the far user while improving the near user's rate. To enhance the performance of the QoS provisioning AS scheme, the idea of dynamic antenna clustering is applied at the relay to adaptively partition the relay's antennas into transmit and receive subsets. The proposed AS schemes' exact outage probability and achievable rate expressions are derived. To provide more insight, closed-form asymptotic outage probability expressions for the max-$\SUu$ and max-$\SUuu$ AS schemes are obtained. Our results show that while the QoS provisioning AS scheme can deliver a near-optimal performance for static antenna setup at the relay,  it provides up to $12\%$ average sum rate gain over the optimum AS selection with fixed antenna setup.
\end{abstract}
\vspace{-0em}
\vspace{-2em}
\section{Introduction}
Driven by the unprecedented proliferation of versatile mobile devices and an upsurge in the growth of innovative applications, as well as the significant increasing demand for user access required for the Internet of Things (IoT), network developers and operators face challenges to deal with the tremendous data traffic requirements for the fifth generation (5G) networks and beyond. This requires a paradigm shift toward developing critical enabling technologies for future wireless networks. To this end, non-orthogonal multiple access (NOMA) has been well recognized as a   technique to improve spectral efficiency and provide universal connectivity challenges in 5G networks and beyond by  serving multiple users simultaneously over the same channel~\cite{Ding:MAG:2017,Dai:2018}. NOMA contrasts the current orthogonal multiple access (OMA) techniques that allocate one specific resource block to a specific user. In particular, a power-domain   NOMA  system multiplexes signals between users with strong channel conditions (denoted by near users) and users with weak channel conditions (denoted by far users) and uses the successive interference cancellation (SIC) method at the receivers to cancel the inter-user interference due to transmission over the same resource channel~\cite{Saito:VTC2013,Ding:Survay,Zhiguo:CLET:2015}.


In order to improve the reliability/throughput of NOMA systems, the cooperative relaying technique has been proposed to incorporate into NOMA systems, termed as \emph{cooperative NOMA} system. Cooperative NOMA systems are mainly classified into user-assisted cooperative NOMA and relay-assisted cooperative NOMA. In user-assisted cooperative NOMA, the near user assists the far user, exploiting the fact that the near user can decode the information for both users~\cite{Zhiguo:CLET:2015}. However, in the relay-assisted cooperative NOMA, the far user is being helped by a dedicated relay~\cite{Ding:2016:wcoml}. However, half-duplex (HD) relaying in such networks incurs spectral efficiency losses~\cite{Caijun:CLET:2016,Mobini:TWC:2019,Guizani:TCOM:2020,Abderrahmane:TCOM:2021}. Thus, one may recover this spectral loss by exploiting full-duplex (FD) relaying, where the relay node can receive and transmit simultaneously in the same frequency band. FD is also a popular choice for the 5G  wireless networks, and it has been shown in the literature that the marriage between FD and NOMA can boost the spectral efficiency~\cite{Mohammadi:MCOM,Caijun:CLET:2016,Yue:tcom:2017,Zhang:JSAC:2017,Mohammadi:TCOM:2018}. FD operation in a relay-assisted NOMA system brings two new challenges: inter-user interference at the near user and self-interference (SI) at the relay where the former is due to relaying the far user's information, and the later is due to the signal leakage from the relay's output to its input. Nevertheless, many techniques have been recently proposed to effectively suppress the SI~\cite{Sabharwal:TWC2012,Zhang:2015}.

Besides  FD operation and NOMA,  wireless  networks can also utilize multiple-input multiple-output (MIMO) technology to enhance system performance. The application of MIMO technology to NOMA systems has been investigated in~\cite{WCOM:2016,Vaezi:TCCN:2019,Silva:TCOM:2020}. However,  the penalty that is paid is the higher computational complexity as well as  the cost of hardware radio frequency chains,  which scales with the number of antennas~\cite{Molisch}. Antenna selection (AS) overcomes  this problem and  has been widely touted as a practical alternative, which reduces  the hardware overhead of traditional MIMO systems while  maintaining comparable performance.
As such, few recent works have investigated the application of AS in NOMA networks~\cite{Chen:CLET:2017,Lei:Access:2017,Vucetic:TVT:2018,Lei:TVT:2018,Tran:SYSJ:2020,Jaiswal:TVT:2021,Han:iet:2016, daCosta:TGCN:2018,Wang:ICC:2018,Serpedin:TVT:2017,Zhang:Access:2020,Aldababsa:TWC:2021}. However, there are still many challenges ahead for the NOMA systems with AS due to the complicated nature of near/far user performance criterion.

\subsection{Related Works}
AS problem has been investigated in MIMO NOMA~\cite{Chen:CLET:2017,Lei:Access:2017,Vucetic:TVT:2018,Lei:TVT:2018,Tran:SYSJ:2020,Jaiswal:TVT:2021} and cooperative NOMA~\cite{Han:iet:2016, daCosta:TGCN:2018,Wang:ICC:2018,Serpedin:TVT:2017,Zhang:Access:2020,Aldababsa:TWC:2021} systems.
In~\cite{Chen:CLET:2017} a  computationally efficient joint AS algorithm was proposed to maximize the signal-to-noise ratio (SNR) of the secondary user under the condition that the quality-of-service (QoS) of the primary user is satisfied. In~\cite{Lei:Access:2017}, the secrecy outage performance of two-user multiple-input single-output (MISO) NOMA systems with different transmit AS (TAS) schemes has been investigated. The authors in~\cite{Vucetic:TVT:2018} investigated joint transmit and receive AS problem for a two-user MIMO NOMA system for  fixed power allocation NOMA and cognitive radio-inspired NOMA (CR-NOMA) scenarios, where both the base station (BS) and users are equipped with multiple antennas. The TAS problem has been investigated in~\cite{Lei:TVT:2018} and~\cite{Tran:SYSJ:2020} to improve the secrecy performance of two-user and multi-user downlink MIMO-NOMA systems, respectively. The authors in~\cite{Jaiswal:TVT:2021}, derived the exact expressions for the outage probability, average bit error rate, and ergodic capacity of the MISO NOMA-based vehicular systems over double Nakagami-$m$ fading channels.

The outage performance of a cooperative NOMA system with a multiple-antenna energy harvesting relay has been analyzed
in~\cite{Han:iet:2016}, where  the transmit antenna that maximizes the instantaneous channel gain between the BS
and  relay  is selected at the BS.
In order to improve the performance of cell-edge users in two-user MISO-NOMA cooperative downlink transmissions, three cooperative schemes were proposed in~\cite{daCosta:TGCN:2018}, where different transmit AS criteria were adopted to provide flexible system design choices. In~\cite{Wang:ICC:2018}, the study in~\cite{Chen:CLET:2017} was extended to investigate the problem of joint relay and AS strategy for amplify-and-forward CR-NOMA systems. In~\cite{Serpedin:TVT:2017}, the outage performance of a NOMA-based downlink amplify-and-forward relay network over Nakagami-$m$ fading channels with imperfect channel state information (CSI) has been studied, where the BS and all users are equipped with multiple antennas. The authors in~\cite{Zhang:Access:2020} proposed a low-complexity relay and AS scheme for a wireless powered cooperative NOMA system, where the proposed AS scheme at the BS aims to maximize the harvested energy of the BS. Outage probability performance of a cooperative MIMO-NOMA system with the presence and absence of channel estimation errors and feedback delays has been examined in~\cite{Aldababsa:TWC:2021}.

Until now, a few previous works have tackled the combination of  AS with FD multi-antenna cooperative NOMA system~\cite{Mohammadi:icc:2018,Pei:TVT:2020}. Specifically, ~\cite{Mohammadi:icc:2018} develops several AS schemes for FD cooperative NOMA systems, aiming at maximizing the received signal-to-interference noise ratio (SINR) at near or far NOMA users. The authors in~\cite{Pei:TVT:2020} investigated the secrecy outage probability of an FD relay-assisted cooperative NOMA system under eavesdropping attack, where the BS and relay adopt AS to enhance the security performance.
However, best of the authors' knowledge, no work has addressed the AS problem in FD cooperative NOMA networks by taking the QoS requirement of the NOMA users into consideration.

\subsection{Contributions and Outcomes}
This paper develops three novel AS schemes in an FD cooperative NOMA system where a multi-antenna BS communicates with two single-antenna users using a multi-antenna FD relay. The fundamental challenge behind the integration of AS and FD NOMA
is the problem of  the SI~\cite{Sabharwal:TWC2012} and interference between the relay and near user~\cite{Guizani:TCOM:2020}, which should be considered during the selection of strong channels toward the near/far users. This makes the AS problem in FD NOMA systems much more challenging than in HD NOMA systems~\cite{Vucetic:ICC17}. Specifically, for the considered FD NOMA relay system with fixed (static) receive/transmit antenna configuration, we propose two low-complexity AS schemes to achieve near/far user end-to-end (\emph{e2e}) SINR maximization. Moreover, we design a QoS provisioning AS scheme, which ensures the far user's targeted data rate is realized while serving the near user with a rate as large as possible.
Our contributions can be summarized as follows:

\begin{itemize}
\item The performance of the FD cooperative NOMA system with the proposed  AS schemes is analyzed by deriving exact outage probability and achievable rate expressions of the NOMA far and near user. In addition, we present simple high-SNR approximations for the outage probability of the near and far users with max-$\SUu$, max-$\SUuu$ AS schemes. These results enable us to gain valuable insights on the impact of key system parameters on the outage performance  achieved by the proposed AS schemes.
\item To enhance the achievable rate performance of the QoS AS scheme, dynamic antenna clustering at the FD relay is considered. Specifically, the antennas at the relay are adaptively configured for reception and transmission, leading to more degrees of freedom for AS design. QoS provisioning AS scheme with both dynamic  antenna clustering and static antenna setup lead to efficient use of available radio resources by providing more degrees-of-freedom for interference reduction at the near user while ensuring the far user target rate.

\item Our findings reveal that QoS provisioning AS scheme with static antenna setup ensures the far user's target rate and provides the near-optimum average sum rate in the entire SNR range compared to the other proposed AS schemes. Moreover, the performance of the QoS provisioning AS scheme with dynamic antenna clustering is significantly improved, providing up to $12\%$ gain over the optimum AS scheme with static antenna setup at the relay.
\end{itemize}

\vspace{-0.5em}
\subsection{Organization and Notations}
The rest of the paper is organized as follows. In Section~\ref{sec:sys}, the system model and proposed AS schemes are presented. In Section~\ref{sec:perf}, the performance of the proposed AS schemes in terms of outage probability and the achievable rate is investigated. In Section~\ref{sec:DQoS}, QoS provisioning AS with dynamic antenna clustering is presented. In Section~\ref{sec:imp}, we elaborate on implementation details of the proposed AS schemes. Numerical results are presented in Section~\ref{sec:Num}, followed by conclusions in Section~\ref{sec:Conc}.

\begin{figure}
\centering
\vspace{0em}
\includegraphics[width=200mm, height=240mm]{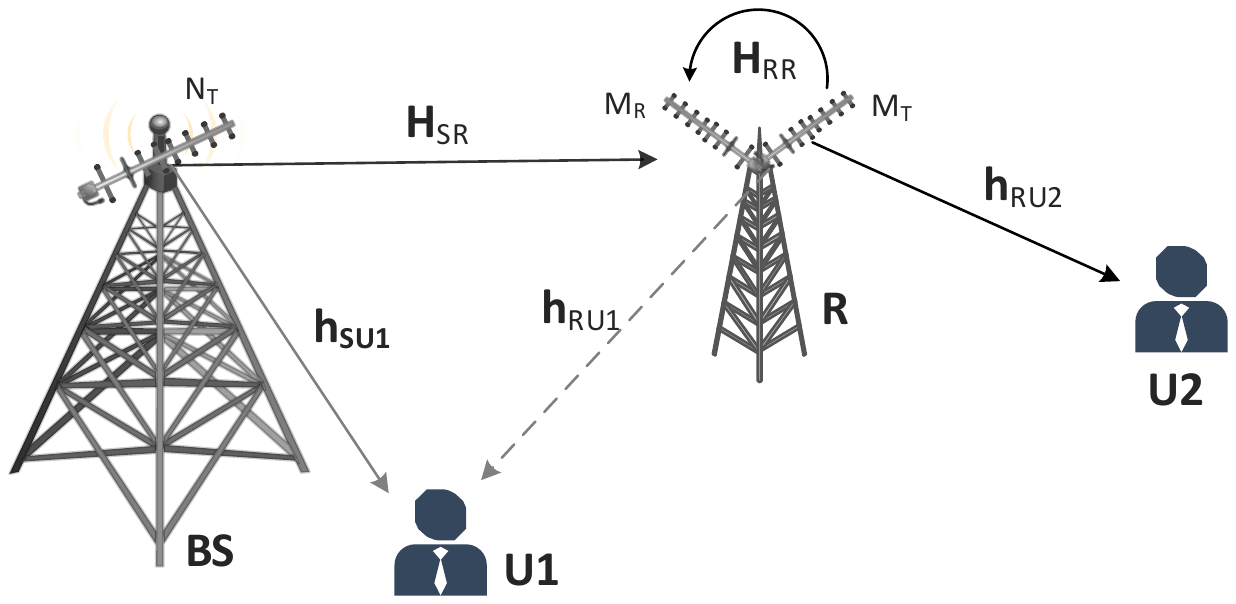}
\vspace{-43em}
\caption{ FD Cooperative NOMA system with antenna selection.}\label{fig:model}
\vspace{-1em}
\end{figure}

\emph{Notation:} We use $\mathbb{E}\left\{X\right\}$ to denote the expected value of the random variable (RV) $X$; its  probability density function (pdf) and  cumulative distribution function (cdf) are $f_X(\cdot)$ and $F_X(\cdot)$  respectively; $\mathcal{CN}(0,\sigma^2)$ denotes a circularly symmetric complex Gaussian RV $X$ with variance $\sigma^2$ and $\mathrm{E_i}(x)=\int_{-\infty}^{x}\frac{e^t}{t}dt$ is the exponential integral function~\cite[Eq. (8.211.1)]{Integral:Series:Ryzhik:1992}.
\vspace{-0.5em}
\section{System Model}\label{sec:sys}
We consider a downlink NOMA system consisting of a  BS, an FD relay, $\Rly$, and two NOMA users,  $\SUu$ (near user) and $\SUuu$ (far user) as shown in Fig.~\ref{fig:model}. The BS is equipped with $\MB$ antennas, $\Rly$ has an FD transceiver equipped with $M = \MR +\MT$ antennas of which $\MR$ antennas are devoted to receiving signals, and $\MT$ antennas are used for data transmission, while $\SUu$ and $\SUuu$ are equipped with a single antenna each. We assume that $\SUu$ communicates directly with the BS, while $\SUuu$ requires the assistance of $\Rly$  employing the decode-and-forward (DF) relaying protocol.

In order to reduce the implementation complexity of the system, we assume that a single AS is performed at the BS and $\Rly$~\cite{Molisch}. In particular, among $\MB$ available transmit antennas, the BS selects one antenna, e.g., $i$-th,   to perform downlink transmission. In addition, among $\MR$ ($\MT$)  available receive (transmit) antennas $\Rly$ selects one, e.g. $j$-th ($k$-th),   to receive signals (forward the BS signal to $\SUuu)$.

We assume that all wireless links  experience non-selective independent Rayleigh fading. The link corresponding to the $j$-th receive and the $i$-th transmit antenna from node $X$ to node $Y$, is denoted by $h_{XY}^{i,j}\sim\mathcal{CN}(0,\SXY)$ where $X \in\{\Srs,\Rly\}$ and $Y \in\{\Rly, \Nu,\Fu\}$.

Let $x_u$, $u\in\{1,2\}$ denotes the information signal intended for $\text{U}u$, and $a_u$ denotes the NOMA power allocation coefficient with  $a_1 + a_2 =1$ and $a_1< a_2$. Based on NOMA principle~\cite{Ding:Survay,Zhiguo:CLET:2015}, the superimposed
signals at the BS is given by
\vspace{-0.8em}
\begin{align}
s[n]=\sqrt{\PS a_1}x_1[n]+\sqrt{\PS a_2}x_2[n],
\end{align}
where $\PS$ is the BS transmit power.  Since, $\Rly$ is equipped with FD transceiver, it receives and transmits on the same channel. Hence, the received signal at $\Rly$ can be written as
\vspace{-0.3em}
\begin{align}
y_{\mathtt{R}}[n]&= \HBRij s[n]+\sqrt{\PR}\HSIkj x_2[n-\tau]+ n_{\mathtt{R}}[n],
\end{align}
where $\PR$ and $x_2[n-\tau]$  are the relay transmit power and relay transmit signal, respectively.   $\tau\geq1$ and $n_{\mathtt{R}}[n]\sim\mathcal{CN}(0,\Sn)$ denote the processing delay~\cite{Riihonen:JSP:2011} and   the additive white Gaussian noise (AWGN)  at  $\Rly$. Also, Rayleigh flat fading, a well accepted model in the literature~\cite{Riihonen:JSP:2011}, is used to model the  $\MR\times\MT$ SI channel $\HSI =\Big[\HSIkj\Big]_{k=1,\ldots, \MT}^{j=1,\ldots, \MR}$. Therefore, $\HSIkj$ is modeled as independent identically distributed (i.i.d.) $\mathcal{CN}(0,\Sap)$ RV.

Applying SIC and treating the symbol of $\SUu$ as interference~\cite{Caijun:CLET:2016}, $\Rly$ decodes the information intended for $\SUuu$. Therefore, the SINR at $\Rly$ can be expressed as
\vspace{-0.3em}
\begin{align}\label{eq:SINR relay}
\GMR=\frac{ a_2\GBRij}{a_1\GBRij \!+\!\! \GSIkj\!+1},
\end{align}
where $\GBRij=\rho_S|\HBRij|^2$ and $\GSIkj=\rho_R|\HSIkj|^2$ with $\rho_S=\frac{\PS}{\Sn}$ and $\rho_R=\frac{\PR}{\Sn}$.

Meanwhile, $\SUu$ receives two signals, one is the transmit signal from BS and the other is the inter-user interference caused by the FD operation at $\Rly$. Thus, the observation at $\SUu$ can be expressed as
\vspace{-0.5em}
\begin{align}\label{eq:sig:U1}
y_1[n]\!=\!\HBNu s[n]\!+\!\sqrt{\PR} \HRNu x_2[n\!-\!\tau]\!+\!n_1[n],
\end{align}
where $n_1[n]\sim\mathcal{CN}(0,\Sn)$ is the AWGN at $\SUu$.

Based on the NOMA principle, $\SUu$  employees SIC to fist decode  $x_2[n]$ of $\SUuu$, who has less interference. Then, after removing $x_2[n]$, $\SUu$ obtains its intended symbol $x_1[n]$.  Thus, the received SINR at $\SUu$ to detect $x_2[n]$ is given by
\vspace{-0.5em}
\begin{align}\label{eq:SINR UE2 at UE1}
\gamma_{12}&=\frac{ a_2\GSNu}{ a_1\GSNu +\GRNu+1},
\end{align}
where $\GSNu = \rho_S|\HBNu|^2$ and  $\GRNu=\rho_R|\HRNu|^2$. According to SIC at $\SUu$, it is rational to assume that the symbol, $x_{2}[n-\tau]$, is priory known to $\SUu$ and hence, $\SUu$ can  remove it~\cite{Caijun:CLET:2016}. Nevertheless, here, we consider realistic imperfect interference cancellation assumption wherein $\SUu$ cannot completely remove $x_{2}[n-\tau]$, due to the unknown interference channel between $\Rly$ and $\SUu$. We model $\HRNu \sim \mathcal{CN}(0,k_1\sigma^2_{\mathsf{RU1}})$ as the inter-user interference channel between $\Rly$ and $\SUu$, where the parameter $k_1$ controls the strength of the inter-user interference~\cite{Caijun:CLET:2016}. If $\SUu$ completely cancels the $\SUuu$'s signal, the SINR at $\SUu$ is given by
\vspace{-0.4em}
\begin{align}\label{eq:SINR at UE1}
\gamma_{1}&=\frac{ a_1\GSNu}{\GRNu+1}.
\end{align}

Relay forwards $x_2[n]$ to the $\SUuu$. Accordingly,  the received signal at $\SUuu$, from $\Rly$ can be expressed as
\begin{align}
y_2[n]=\sqrt{\PR}  \HRFu x_2[n-\tau]+n_2[n],
\end{align}
where $n_2[n]\sim\mathcal{CN}(0,\Snuu)$ denotes the AWGN at $\SUuu$. Hence, the received SNR at $\SUuu$ is given by
\vspace{-0.5em}
\begin{align}\label{eq:gamR2}
\GRFu=\frac{\PR}{\Snuu }|\HRFu|^2.
\end{align}

In order to obtain the \emph{e2e} SINR at $\SUuu$, we notice that the weakest link constraints the achievable rate of dual-hop DF relay network, ie., $\log_2(1+\min(\GMR,\GRFu))$.  Moreover, $\SUu$ must be able to  decode the intended signal for $\SUuu$~\cite{Caijun:CLET:2016}. Therefore,  \emph{e2e} SINR at $\SUuu$ can be expressed as
\begin{align}\label{eq:e-2-e far}
\gamma_{2}&=\min\!\left(\!\frac{ a_2\GSNu}{ a_1\GSNu+\GRNu+1},\frac{a_2\GBRij}{ a_1\GBRij+\GSIkj+1},
\GRFu\right).
\end{align}

\subsection{Antenna Selection Schemes}
In this subsection, we propose three AS schemes, called max-$\SUu$, max-$\SUuu$, and QoS provisioning, for the considered FD cooperative NOMA network where a joint selection of  transmit and receive antennas is performed at the BS and relay. More specifically, in the  max-$\SUu$   and max-$\SUuu$ AS scheme,  one single antenna from the BS and one pair of a receive-transmit antenna at the relay are selected based on the received \emph{e2e} SINRs at the near user, $\SUu$, and far user, $\SUuu$, respectively. We further propose a two-stage  QoS provisioning AS scheme to maximize the received SINR of $\SUu$ while guaranteeing a specific target rate at $\SUuu$.

\subsubsection{max-$\SUu$ AS Scheme}
In this one,  one transmit antenna at the BS and one transmit antenna at the relay are first selected such that the   \emph{e2e} SINR at $\SUu$, given in ~\eqref{eq:SINR at UE1}, is maximized. Then,  with remaining AS choices, one receive antenna at the relay is selected to maximize the \emph{e2e} SINR at $\SUuu$. Thus, max-$\SUu$ AS can be mathematically written as
\begin{align}\label{eq:ASc near}
\{i^*,  k^*\} =\argmax_{\substack{1\leq i\leq\MB,  1\leq k\leq\MT}}\frac{ a_1\GSNu}{\GRNu+1},\qquad
j^*=\argmax_{\substack{ 1\leq j\leq\MR}}\frac{a_2\GBRiij}{ a_1\GBRiij + \GSIkkj+1}.
\end{align}
\subsubsection{max-$\SUuu$ AS Scheme}
Here, one single antenna from the BS and one pair of receive-transmit antenna at the relay are selected to maximize the \emph{e2e} SINR at $\SUuu$ according to
\vspace{-0.1em}
\begin{align}\label{eq:ASc far}
\{i^*, j^*, k^*\} &= \argmax_{\substack{1\leq i\leq\MB, 1\leq j\leq\MR,\\ 1\leq k\leq\MT}}
\min\left(\frac{ a_2\GSNu}{ a_1\GSNu \!+\GRNu\!+1},\frac{a_2\GBRij}{ a_1\GBRij + \GSIkj+1},
\GRFu\right).
\end{align}
Based on~\eqref{eq:ASc far} all degree-of-freedoms in terms of AS (i.e., BS transmit antenna, and relay's receive and transmit antenna) are used to maximize the \emph{e2e} at $\SUuu$, hence there is no  degree-of-freedom available to maximize the \emph{e2e} SINR at $\SUu$. Therefore, under similar conditions the received SINR at the near user in max-$\SUu$ AS scheme outperforms the received SINR at the near user in max-$\SUuu$ AS scheme.

\subsubsection{QoS Provisioning AS Scheme}
In the aforementioned AS schemes, either the $\SUu$'s SINR or $\SUuu$'s SINR is maximized without considering the other user's QoS requirement. The main object of this AS scheme is to simultaneously ensure $\SUuu$'s targeted data rate and serve $\SUu$ with a rate as large as possible. Therefore, the QoS provisioning AS scheme is a two-stage AS, described in more detail as follows.\\\\

In the first stage, the following set of all antenna's combinations $\{i, j, k\}$ is created:
\vspace{-0.2em}
\begin{align}\label{eq:TwoStage:far}
\mathcal{A}&= \Bigg\{ \{i, j, k\}: \min\Bigg(\frac{ a_2\GSNu}{ a_1\GSNu \!+\GRNu\!+1},\frac{a_2\GBRij}{ a_1\GBRij + \GSIkj+1}
\GRFu\Bigg)\geq2^{\mathcal{R}_2}-1\Bigg\},\nonumber\\
&\hspace{2em} \forall\quad 1\leq i\leq\MB, 1\leq j\leq\MR, 1\leq k\leq\MT,
\end{align}
where $\mathcal{R}_2$ is the target transmission rate at $\SUuu$. Therefore, if $\mathcal{A}$ is non empty, the intended signal for $\SUuu$ can be decoded at the relay while the intended signal for $\SUu$ is treated as noise. In the second stage, antenna set $\{i^*, j^*, k^*\}\in\mathcal{A}$ which maximizes the rate for $\SUu$ is selected as follows. By invoking the received SINR at $\SUu$, given in~\eqref{eq:SINR at UE1}, the best transmit antennas at the BS and Relay are selected such that $\gamma_{1}$ is maximized, i.e.,
\vspace{-1em}
\begin{align}\label{eq:TwoStage:near}
\{i^*, k^*\} &= \argmax_{\mathcal{A}}\left\{\frac{ a_1\GSNu}{\GRNu+1}, \quad \{i, j, k\} \in \mathcal{A}\right\}.
\end{align}
Accordingly, for the given $i^*$, $k^*$ combination in~\eqref{eq:TwoStage:near}, the received antenna at $\Rly$ is selected, such that the received SINR at $\Rly$ is maximized, i.e.,
\vspace{-0.2em}
\begin{align}\label{eq:TwoStage:far}
j^*=&\argmax_{\substack{\mathcal{A} }}
\left\{\frac{a_2\GBRiij}{ a_1\GBRiij + \GSIkkj+1},\quad \{i^*, j, k^*\} \in \mathcal{A} \right\}.
\end{align}

We note here that based on~\eqref{eq:TwoStage:far}, one receive antenna at $\Rly$ is selected from those antenna combinations that lies in $\mathcal{A}$ and maximizes the rate for $\SUu$.

Section~\ref{sec:DQoS} describes the QoS provisioning AS scheme with dynamic antenna clustering in more detail.

\section{Performance Evaluation}~\label{sec:perf}
This section investigates the system performance with the proposed AS schemes. We focus on two information-theoretic criteria, namely outage probability and achievable rate, and  derive closed-form  expressions for evaluating and comparing the proposed AS schemes.

\vspace{-0.3em}
\subsection{Outage Probability Analysis}
Outage probability is an effective performance measure of the communication systems operation over fading channels. The outage probability is defined as the event that the data rate supported by instantaneous channel realizations is less than a targeted user rate.

In order to derive the outage probability expressions for the max-$\SUu$ and max-$\SUuu$ AS schemes, we notice that an outage event at $\SUu$ occurs in one of the following cases: \emph{1)} the intended information signal for $\SUuu$ cannot be decoded correctly by $\SUu$   \emph{2)} $\SUu$  can decode it but fails to decode its signal. Thus, the $\SUu$ outage performance of can be expressed as
\vspace{-0.3em}
\begin{align}\label{eq:outnear}
\Poutnu&=1-\Prob\left(\gamma_{12,\mathtt{AS}}>\theta_2,
\gamma_{1,\mathtt{AS}}>\theta_1\right)= F_{\gamma_{1,\mathtt{AS}}}\left(\zeta\right),
\end{align}
where $\gamma_{1,\mathtt{AS}}=\frac{\GSNu}{\GRNu+1}$, $\zeta={ \max\left(\cno,\frac{\theta_1}{a_1} \right)}$ with $\cno=\frac{\theta_2}{a_2-a_1\theta_2}$, $\theta_1 = 2^{\mathcal{R}_1}-1$, $\theta_2 = 2^{\mathcal{R}_2}-1$, $\mathcal{R}_1$ is the transmission rates at $\SUu$, and $\mathtt{AS}\in\{\Ss, \Sss\}$.\footnote{Hereafter, we use the superscripts/subscripts $\mathtt{AS}\in\{\Ss, \Sss, \Ssss\}$ in related variables to denote the max-$\SUu$, max-$\SUuu$ and QoS provisioning AS schemes, respectively.}

Moreover, the outage event at $\SUuu$ with max-$\SUu$ and max-$\SUuu$ AS schemes  is due to the following cases: 1)  $\Rly$ cannot correctly decode the intended message to $\SUuu$, \emph{2)}  $\SUuu$ signal can be decoded by $\Rly$   but $\SUuu$ fails to decode its own information. Therefore, the outage probability of $\SUuu$ can be expressed as
\vspace{-0.3em}
\begin{align}\label{eq:outfar}
\Poutfu
&=1-\Prob\left(\GMRi>\theta_2\right)\Prob\left(
\GRFui>\theta_2\right).
\end{align}

Now we characterize the outage probability of both user with max-$\SUu$, max-$\SUuu$, and QoS provisioning AS schemes.

\begin{proposition}~\label{outageU1:S1S2}
The outage probability of $\SUu$ with max-$\SUu$ AS scheme,  denoted by $\PoutnuS$,  and the outage probability of max-$\SUuu$ AS scheme, denoted by $\PoutnuSS$, are respectively given by
\vspace{-0.3em}
\begin{align}\label{eq:outNEARS1}
&\PoutnuS=1- \!\MB\sum_{p=0}^{\MB-1}
\frac{ (-1)^p\binom{\MB-1}{p} e^{-\frac{(p+1)\zeta}{\bGSNu}}}
{(p+1)\left(1 +\frac{\bGRNu}{\bGSNu}\frac{(p+1) \zeta}{\MT}\right)},
\end{align}
and
\vspace{-1.1em}
\begin{align}\label{eq:outNEARS2}
&\PoutnuSS
=1- \frac{e^{-\frac{\zeta}{\bGSNu}}}{1+\frac{\bGRNu}{\bGSNu}\zeta},
\end{align}
where   $\bGSNu=\rho_S\SBNu$ and $\bGRNu=\rho_R\SRNu$.
\end{proposition}

\begin{proof}
See Appendix~\ref{Appendix:outageU1:S1S2}.
\end{proof}

\begin{Remark}
Proposition~\ref{outageU1:S1S2} indicates that the outage probability of $\SUu$ with  max-$\SUu$ scheme decays by increasing  $\MB$ and $\MT$. In contrast, the outage probability of $\SUu$ with the max-$\SUuu$ scheme is independent of the number of antennas at the BS and relay. This can be interpreted as the max-$\SUuu$ scheme utilizing all degree-of-freedom to maximize the \emph{e2e} at the $\SUuu$, and hence no degree-of-freedom is remained to maximize the \emph{e2e} SINR at $\SUu$. Therefore, max-$\SUuu$ reduces to the random AS scheme, from the $\SUu$'s point of view.
\end{Remark}

\begin{proposition}~\label{outageU2:S1S2}
The outage probability of $\SUuu$ with max-$\SUu$ AS scheme, $\PoutfuS$,  and the outage probability of $\SUuu$ with max-$\SUuu$ AS scheme,   $\PoutfuSS$,  for $\frac{a_2}{a_1}>\theta_2$, are respectively given by
\begin{align}
\PoutfuS&=1-\MR e^{-\frac{\theta_2}{\bGRFu}}\sum_{q=0}^{\MR-1}\!\frac{(-1)^q\binom{\MR-1}{q}
e^{-\frac{(q+1)\cno}{\bGBR}}}
{(q+1)\left(1+\!\frac{\bGSI}{\bGBR}(q+1)\cno\right)},
\label{eq:outnuS2}
\end{align}
and
\vspace{-1em}
\begin{align}
\PoutfuSS&=1-\MT\MB\!\sum_{q=0}^{\MT-1}
\frac{ (-1)^q\binom{\MT-1}{q} e^{-\frac{(q+1)\theta_2}{\bGRFu}}}
{(q+1)}\sum_{p=0}^{\MB-1}\frac{(-1)^p\binom{\MB-1}{p}e^{-\frac{(p+1)\cno}{\bGBR}}}
{(p+1)\left(1+\frac{\bGSI}{\bGBR}\frac{(p+1)\cno}{\MR}\right)}\label{eq:outfuS2},
\end{align}
where  $\bGBR = \rho_S\SBR$, $\bGRFu=\rho_R\SRFu$ and $\bGSI=\rho_R \Sap$.
\end{proposition}

\begin{proof}
See Appendix~\ref{Appendix:outageU2:S1S2}.
\end{proof}

Proposition~\ref{outageU2:S1S2} presents the  impact of different system parameters on the outage performance of  $\SUuu$ with max-$\SUu$  and max-$\SUuu$ AS schemes. Specifically, it reveals that $\PoutfuS$ only depends on the number of relay's receive antennas, $\MR$,  while $\PoutfuSS$ is a function of the number of antennas at both BS and FD relay.

To gain further insights on the diversity order of the max-$\SUu$  and max-$\SUuu$ AS schemes, we now look into the high-SNR regime and derive simple approximations for the outage probability of $\SUu$ and $\SUuu$ with these schemes.

\begin{Corollary}~\label{Corr:S1S2}
For any $\rho_S$ and $\rho_R$  such that $\frac{\rho_R}{\rho_S} =c_1$ is fixed, as $\rho_R$, $\rho_S\rightarrow\infty$, the outage probability of $\SUu$ with max-$\SUu$ and max-$\SUuu$ AS scheme, can be respectively approximated as
\vspace{-0.1em}
\begin{align}\label{eq:outNEARS1:APX}
&\PoutnuS \approx
\frac{\Gamma(\MB+1)}{\MT^{\MB}}
\left(\frac{c_1\SRNu}{\SBNu}\zeta\right)^{\!\!\MB},
\end{align}
and
\vspace{-1em}
\begin{align}\label{eq:outNEARS2:APX}
&\PoutnuSS
\approx\frac{\frac{\bGRNu}{\bGSNu}\zeta}{1+\frac{\bGRNu}{\bGSNu}\zeta}\left(1 + \frac{1}{\bGRNu}\right)
\approx\frac{\zeta}{\frac{\SBNu}{c_1\SRNu}+\zeta}.
\end{align}
Moreover, the outage probability of $\SUuu$ with max-$\SUu$ and max-$\SUuu$ AS scheme, can be respectively approximated as
\vspace{-0.4em}
\begin{align}
\PoutfuS&\approx
\Gamma(\MR+1)\left(\cno\frac{c_1\Sap}{\SBR}\right)^{\!\!\MR},
\label{eq:outnuS2:Apx}
\end{align}
and
\vspace{-1em}
\begin{align}
\PoutfuSS&\approx
  \Gamma(\MB+1)\left(\frac{\cno}{\MR}\frac{c_1\Sap}{\SBR}\right)^{\!\!\MB}.
\label{eq:outfuS2:Apx}
\end{align}
\end{Corollary}

\begin{proof}
The proof is straightforward and thus omitted for the sake of brevity.
\end{proof}

\begin{Remark}
By inspecting Corollary~\ref{Corr:S1S2}, we find  that max-$\SUu$ and max-$\SUuu$ AS scheme suffer from zero-order diversity for both users. This is intuitive due to the presence of the inter-user interference and SI at the relay, which cannot be completely eliminated through the AS selection. Moreover,  max-$\SUu$  AS scheme provides lower floor for $\SUu$, while max-$\SUuu$ AS scheme results in lower floor for $\SUuu$ for the same number of BS's transmit antenna and relay's receive antenna.
\end{Remark}

In what follows, we provide key results for the outage probability of the  QoS provisioning AS Scheme with fixed (static) antenna setup at $\Rly$.

\begin{proposition}~\label{prop:outage:S3}
For the QoS provisioning AS scheme, the outage probability of $\SUu$, $\PoutnuSSS$, and the outage probability of $\SUuu$,  $\PoutfuSSS$, are given by
\vspace{-0.3em}
\begin{align}\label{eq:Pout:U1:Qos}
\PoutnuSSS&=(\Pkk)^{\MT} +
\sum_{r=1}^{\MT}\binom{\MT}{r}
(\Pkk)^{\MT-r} (1-\Pkk)^{r}
(\Pii)^{\MB}\nonumber\\
&\hspace{1em}+
\sum_{r=1}^{\MT}
\sum_{t=1}^{\MB}
\binom{\MT}{r}
\binom{\MB}{t}
(\Pkk)^{\MT-r} (1-\Pkk)^{r}
(\Pii)^{\MB-t} (1-\Pii)^{t}
(\Pjj)^{t\MR}\nonumber\\
&\hspace{1em}+
\sum_{r=1}^{\MB\MT}
\!
\sum_{t=1}^{\MB}
\!
\sum_{\substack{s=1 \\ s \times t = r \\ \text{or}~ s=t=r}}^{\MT}
\!
\binom{\MB\MT}{r}(\Pll)^r (1\!-\Pll)^{\MB\MT-r}
\Pmm,
\end{align}
and
\vspace{-1em}
\begin{align}\label{eq:Pout:U2:Qos}
\PoutfuSSS= &(\Pkk)^{\MT} +
\sum_{r=1}^{\MT}\binom{\MT}{r}
(\Pkk)^{\MT-r} (1-\Pkk)^{r}
(\Pii)^{\MB}\nonumber\\
&\hspace{2em}+
\sum_{r=1}^{\MT}
\sum_{t=1}^{\MB}
\binom{\MT}{r}
\binom{\MB}{t}
(\Pkk)^{\MT-r} (1-\Pkk)^{r}
(\Pii)^{\MB-t} (1-\Pii)^{t}
(\Pjj)^{t\MR},
\end{align}
respectively, where
\vspace{-1em}
\begin{subequations}
\begin{align}
\Pii &=1-e^{-\frac{\cno}{\bGSNu}} \left(\frac{\cno \bGRNu}{r\bGSNu}+1\right)^{-1},\label{eq:Pii}\\
\Pjj &=1-e^{-\frac{\cno}{\bGBR}} \left(\frac{\cno\bGSI}{\bGBR}+1\right)^{-1},\label{eq:Pjj}\\
\Pkk &=1 - e^{-\frac{\theta_2}{\bGRFu}},\label{eq:Pkk}\\
\Pll &=\frac{e^{-\left(\frac{\cno}{\bGSNu}+\frac{\cno}{\bGBR}+\frac{\theta_2}{\bGRFu}\right)}}
{\bGRFu\left(\frac{\cno\bGRNu}{\bGSNu}+1\right)\left(\frac{\cno\bGSI}{\bGBR}+1\right)},\label{eq:Pll}\\
\Pmm &=1\!-\! t\sum_{p=0}^{t-1}
\frac{ (-1)^p\binom{t-1}{p} e^{-\frac{(p+1)\zeta}{\bGSNu}}}
{(p\!+\!1)\left(1 \!+\! \frac{(p+1)\bGRNu \zeta}{s \bGSNu}\right)}\label{eq:Pmm}.
\end{align}
\end{subequations}
\end{proposition}

\begin{proof}
We first derive the outage probability of $\SUuu$, since the outage probability of $\SUu$ is defined based on $\SUuu$. An outage at $\SUuu$ occurs if set $\mathcal{A}$ is empty. By invoking~\eqref{eq:TwoStage:far}, it is clear that set $\mathcal{A}$ is empty if the following conditions are satisfied: 1) $\mathcal{O}_1$: when there is no active link between $\Rly$ and $\SUuu$, i.e., $\GRFu<2^{\mathcal{R}_2}-1$, $\forall~1\leq k\leq\MT$, 2) $\mathcal{O}_2$: when there is $1\leq k\leq\MT$ transmit antenna at the relay which guarantees $\GRFu\geq2^{\mathcal{R}_2}-1$, but there is not any transmit antenna at the BS, such that antenna combination $\{i,k\}$ satisfies $ \frac{a_2\GSNu}{ a_1\GSNu+\GRNu+1}\geq2^{\mathcal{R}_2}-1$, $\forall 1<i<\MB$, 3) $\mathcal{O}_3$: when there is at least one transmit antenna at the BS and relay, but there is not any active BS-$\Rly$ link to satisfy $\frac{a_2\GBRij}{ a_1\GBRij+\GSIkj+1}\geq2^{\mathcal{R}_2}-1$. Therefore, the outage probability of $\SUuu$ can be expressed as
\vspace{-0.6em}
\begin{align}\label{eq:Pout:U2:def}
\PoutfuSSS=\Prob(\mathcal{O}_1)+\Prob(\mathcal{O}_2)+\Prob(\mathcal{O}_3).
\end{align}
Now, we proceed to derive three probabilities on the right hand side of~\eqref{eq:Pout:U2:def}. Considering the fact that all $\Rly$-$\SUuu$ links are independent, $\Prob(\mathcal{O}_1)$ can be readily obtained as
 \vspace{-0.4em}
\begin{align}\label{eq:Pro1}
\Prob(\mathcal{O}_1) &= \Pi_{k=1}^{\MT}\Prob\left(\GRFu<2^{\mathcal{R}_2}-1\right)\nonumber\\
&= \left(1 - e^{-\frac{\theta_2}{\bGRFu}}\right)^{\MT} = (\Pkk)^{\MT}.
\end{align}

In order to derive $\Prob(\mathcal{O}_2)$, we notice that based on the proposed QoS provisioning AS scheme, from $1\leq k\leq\MT$ available $\Rly$-$\SUuu$ links, the weakest one is selected. Therefore, $\Pii$ can be obtained as
 \vspace{-0.3em}
\begin{align}\label{eq:Pro2}
\Prob(\mathcal{O}_2) = \sum_{r=1}^{\MT}\Prob(k=r)
\left(\underbrace{\Prob\Bigg(\frac{a_2\GSNu}{ a_1\GSNu+\min_{1\leq k \leq r}\GRNu +1}
\leq2^{\mathcal{R}_2}-1\Bigg)}_{\triangleq \Pii}\right)^{\MB},
\end{align}
where $\Prob(k=r) = \binom{\MT}{r}\Pkk^{\MT-r}(1-\Pkk)^{r}$. Next, by utilizing the pdf of $A_3 \triangleq\min_{1\leq k \leq r}\GRNu$, i.e., the minimum of $1\leq k \leq r$ exponentially distributed RVs with parameter $\bGSNu$, and the cdf of exponential RV $\GSNu$ with parameter $\bGSNu$, we can obtain $\Pii$ as
 \vspace{-0.0em}
\begin{align}\label{eq:Pii}
\Pii
&=1-e^{-\frac{\cno}{\bGSNu}} \left(\frac{\cno \bGRNu}{r\bGSNu}+1\right)^{-1}.
\end{align}

Finally, given that there are $1\leq k\leq\MT$ available $\Rly$-$\SUuu$ link and $1\leq i\leq\MB$ available BS-$\SUu$ links, we have
\vspace{-1.3em}
\begin{align}\label{eq:Pro3}
\Prob(\mathcal{O}_3) = \sum_{r=1}^{\MT}\sum_{t=1}^{\MB}\Prob(k=r)\Prob(i=t)
\left(\underbrace{\Prob\Bigg(\frac{a_2\GBRij}{ a_1\GBRij+\GSIkj+1}\geq2^{\mathcal{R}_2}-1\Bigg)}_{\triangleq \Pjj}\right)^{t\MR},
\end{align}
where $\Prob(i=t) = \binom{\MB}{t}\Pii^{\MB-t}(1-\Pii)^{t}$. Noticing that $\GBRij$ and $\GSIkj$ are exponential RVs with parameters $\bGBR$ and $\bGSI$, respectively, we can reedily obtain $\Pjj$ as~\eqref{eq:Pjj}. To this end, substituting~\eqref{eq:Pro1},~\eqref{eq:Pro2}, and~\eqref{eq:Pro1} into~\eqref{eq:Pout:U2:def}, the desired result in~\eqref{eq:Pout:U2:Qos} is derived.

To derive the outage probability of $\SUu$, according to~\eqref{eq:TwoStage:near}, the received SINR at $\SUu$ is maximized when the strongest BS-$\SUu$ link along with the weakest $\Rly$-$\SUu$ link is selected provided that the links have been already selected in $\mathcal{A}$ according to~\eqref{eq:TwoStage:near}. Therefore, the outage event at $\SUu$ occurs either when antenna set $\mathcal{A}$ is empty or the received SINR at $\SUu$ falls bellow the predefined threshold, i.e.,
\vspace{-0.1em}
\begin{align}\label{eq:outageS3U1}
\PoutnuSSS&= \Prob(|\mathcal{A}|=0) +
\sum_{r=1}^{K_1}\underbrace{\Prob\left(\frac{ \GSNu}{\GRNu+1}<\zeta\bigg| |\mathcal{A}|=r\right)}_{\mathcal{P}_3(r)}
\Prob\left(|\mathcal{A}|=r\right),
\end{align}
where $|\mathcal{A}|$ denotes the cardinality of set $\mathcal{A}$, and $K_1=\MB\MT$. It is clear that $\Prob(|\mathcal{A}|=0) = \PoutfuSSS$.  To calculate~\eqref{eq:outageS3U1}, we now analyze $\mathcal{P}_3(r)$ and $\Prob\left(|\mathcal{A}|=r\right)$ in the following.

Based on the selection criterion of $\mathcal{A}$ in~\eqref{eq:TwoStage:far}, we notice that since only $i$ and $k$ (and not $j$) contribute in $\mathcal{P}_3(r)$, we have $1\leq r\leq K_1$ possible choices in $\mathcal{A}$. Let us define $\mathcal{A}_{r}=\{1\leq i\leq \MB,~1\leq k \leq \MT,~i,k\in\mathcal{A}: k \times i = r,~\text{or}~i=k=r\}$ as subsets of $\mathcal{A}$, containing all possible antenna combinations that contribute in the received SINR at $\SUu$. Therefore, $\mathcal{P}_3(r)$ can be written as
\vspace{-0.4em}
\begin{align}\label{eq:pr:P3}
\mathcal{P}_3(r) &=
\Prob\Bigg( \bigcup_{\substack{\mathcal{A}_{r}}} \frac{ \GSNu}{\GRNu+1}\leq \zeta\Bigg)
\nonumber\\
&=\sum_{i=1}^{\MB}\sum_{\substack{k=1 \\ k \times i = r\\ \text{or}~k=i=r}}^{\MT} \int_{0}^{\infty} F_{A_i}\left((y+1)\zeta\right)f_{B_k}(y)dy,
\end{align}
where $A_i$  is a RV defined as the maximum out of $i$  exponentially distributed independent RVs, while $B_k$ is the minimum out of $k$ exponentially distributed independent RVs. Therefore, using the cdf of $A_i$ and the pdf of $B_k$, after some manipulation, we get
\vspace{-0.3em}
\begin{align}\label{eq:pr:P3lfinal1}
\mathcal{P}_3(r) &=
\sum_{i=1}^{\MB}\sum_{\substack{k=1 \\ k \times i = r\\ \text{or}~i=k=r}}^{\MT}
\left( 1- i\sum_{p=0}^{i-1}
\frac{ (-1)^p\binom{i-1}{p} e^{-\frac{(p+1)\zeta}{\bGSNu}}}
{(p+1)\left(1 + \frac{(p+1)\bGRNu \zeta}{k \bGSNu}\right)}\!\right).
\end{align}

Moreover, we have
\vspace{-0.7em}
\begin{align}\label{eq:pr:A:nnz}
\Prob\left(|\mathcal{A}|=r\right)& =\binom{K_1}{r}
\prod_{n=1}^{K_1-r}\left[1-\Prob\left(\gamma_{12}\geq \theta_2\right)\Prob\left(\GMR\geq \theta_2\right)\Prob\left(
\GRFu\geq \theta_2\right)\right]\nonumber\\
&\hspace{6em}\times\prod_{n=K_1-r+1}^{K_1}\Prob\left(\gamma_{12}\geq \theta_2\right)\Prob\left(\GMR\geq \theta_2\right)\Prob\left(
\GRFu\geq \theta_2\right)\nonumber\\
&=\binom{K_1}{r} \Pll^r (1-\Pll)^{K_1-r},
\end{align}
where $\Pll=\Prob\left(\gamma_{12}\geq \theta_2\right)\Prob\left(\GMR\geq \theta_2\right)\Prob\left(
\GRFu\geq \theta_2\right)$ can be readily obtained considering the fact that all involving channel gains are independent exponential RVs.
To this end, by substituting~\eqref{eq:pr:P3lfinal1} and~\eqref{eq:pr:A:nnz} into~\eqref{eq:outageS3U1} we derive the desired result in~\eqref{eq:Pout:U1:Qos}.
\end{proof}
\begin{Remark}
With QoS provisioning AS scheme, if random AS is performed at the second stage, the outage probability of $\SUu$ reduces to~\eqref{eq:outNEARS2} and the outage probability of $\SUuu$ in~\eqref{eq:Pout:U2:Qos} simplifies as
\begin{align}\label{eq:Pout:U2:Qos:RD}
\PoutfuSSS= &(\Pkk)^{\MT} +
(1-(\Pkk)^{\MT})
(\Pii)^{\MB}
+
(1-(\Pkk)^{\MT})(1 - (\Pii)^{\MB}) (\Pjj)^{\MR\MB}.
\end{align}
\end{Remark}

\subsection{Achievable Rate Analysis }\label{subsec:rate}
Achievable rate is crucial to improving and optimizing while communication networks evolve generation by generation. In this subsection, we study the achievable rate  for the FD cooperative NOMA system with the proposed AS schemes.

Let $\gamma_{\SUui,\mathtt{AS}}$ indicates the $\emph{e2e}$ SINR at the user $\SUui$, $u  \in \{1 , 2\}$ for the  specific AS scheme. The sum achievable rate of the system is written as
\vspace{-0.6em}
\begin{align}\label{eq:sum rate}
{R}_{\mathrm{sum}} =\RnuAS +\RfuAS,
\end{align}
where the achievable rate of $\SUu$ and $\SUuu$ are, respectively, given by
\begin{align}\label{eq:R1R2:cdf}
\RnuAS = \mathbb{E}\left\{\log_2(1+\gamma_{1,\mathtt{AS}})\right\},\quad
\RfuAS = \mathbb{E}\left\{\log_2(1+\gamma_{2,\mathtt{AS}})\right\}.
\end{align}

In the sequel,  analytical achievable rate expressions  at $\SUu$ and $\SUuu$ for the three proposed AS schemes are presented.

\begin{proposition}\label{prop:Ru1u2:S1}
The achievable rates for max-$\SUu$ AS at $\SUu$ and $\SUuu$  can be respectively obtained as
\vspace{-0.1em}
\begin{align}\label{eq:R1:S1}
 &{R}_{\SUu}^{\Ss}= \frac{\MB}{\mathrm{ln} 2}\sum_{p=0}^{\MB-1}
\frac{ (-1)^p\binom{\MB-1}{p} }
{(p+1)\left( \frac{(p+1)\bGRNu }{\MT a_1\bGSNu}-1\right)}\left( e^{\frac{1}{\bGRNu}}\mathrm{E_i}
\left(-\frac{1}{\bGRNu}\right)-e^{\frac{(p+1)}{a_1\bGSNu}}\mathrm{E_i}
\left(-\frac{(p+1)}{a_1\bGSNu}\right)\right)
\end{align}
and
\vspace{-0.2em}
\begin{align}\label{eq:R2:S1}
&{R}_{\SUuu}^{\Ss}\!= \!\frac{\MR\MB}{\mathrm{ln} 2}
\int_{0}^{\frac{a_2}{a_1}}\!\!
\frac{e^{-\frac{x}{\bGRFu}}}{1+x}
\!\sum_{p=0}^{\MB-1}\!
\frac{(-1)^p\binom{\MB-1}{p}e^{-\frac{(p+1)\cnox}{\bGSNu}}}
{(p+1)\left(1\!+\!\frac{\bGRNu}{\bGSNu}\frac{(p+1)\cnox}{\MT}\right)}
\!\sum_{q=0}^{\MR-1}\!
\frac{(-1)^q\binom{\MR-1}{q}e^{-\frac{(q+1)\cnox}{\bGBR}}}
{(q+1)\left(1\!+\!\frac{\bGSI}{\bGBR}(q+1)\cnox\right)} dx.
\end{align}
\end{proposition}

\begin{proof}
See Appendix~\ref{proof:prop:Ru1u2:S1}.
\end{proof}

\begin{proposition}\label{prop:rate:S2}
The achievable rates of $\SUu$ and $\SUuu$ for the max-$\SUuu$ AS scheme can be expressed as
\vspace{-0.2em}
\begin{align}\label{eq:R1:S2}
{R}_{\SUu}^{\Sss}&=\frac{1}{\mathrm{ln} 2}
\frac{a_1\bGSNu}{\left(\bGRNu-a_1\bGSNu\right)}\left( e^{\frac{1}{\bGRNu}}\mathrm{E_i}\left(\!-\frac{1}{\bGRNu}\right)\!-e^{\frac{1}{a_1\bGSNu}}\mathrm{E_i}
\left(\!-\frac{1}{a_1\bGSNu}\right)\right)\!,
\end{align}
and
\vspace{-0.9em}
\begin{align}\label{eq:R2:S2}
&{R}_{\SUuu}^{\Sss}= \frac{\MT\MB}{\mathrm{ln} 2}
\int_{0}^{\frac{a_2}{a_1}}\frac{e^{-\frac{\cnox}{\bGSNu}}}
{\left(1 + \frac{\bGRNu}{\bGSNu} \cnox\right)(1+x)}\sum_{p=0}^{\MB-1}\frac{(-1)^p\binom{\MB-1}{p}e^{-\frac{(p+1)\cnox}{\bGBR}}}
{(p+1)\left(1+\frac{\bGSI}{\MR\bGBR}(p+1)\cnox\right)}\nonumber\\
&\hspace{7em}\times\sum_{q=0}^{\MT-1}
\frac{ (-1)^q\binom{\MT-1}{q} e^{-\frac{(q+1)x}{\bGRFu}}}
{(q+1)}dx,
\end{align}
respectively.
\end{proposition}

\begin{proof}
This proposition can be proofed by following similar steps as in Proposition~\ref{prop:Ru1u2:S1} and thus the proof is omitted.
\end{proof}

\begin{proposition}\label{prop:rate:S3}
The achievable rate of $\SUu$ under QoS provisioning AS scheme, is given by
\begin{align}\label{eq:rate:S3:u1}
{R}_{\SUu}^{\Ssss}&=
\sum_{r=0}^{\MB\MT}
\sum_{i=1}^{\MB}\sum_{\substack{k=1 \\ k \times i = r\\ \text{or}~i=k=r}}^{\MT}
\sum_{p=0}^{i-1}
\binom{\MB\MT}{r}(\Pll)^r (1\!-\Pll)^{\MB\MT-r}
~i (-1)^p\binom{i-1}{p}\nonumber\\
 &\hspace{2em}
\times\frac{1}{ a_1\bGSNu\mathrm{ln} 2}
\int_0^{\infty}
\Bigg(\frac{1}{1 \!+\! \frac{(p+1)\bGRNu}{k a_1\bGSNu}x}
+\frac{\bGRNu}{k\left(1 \!+\! \frac{(p+1)\bGRNu}
{k a_1\bGSNu}x\right)^2}\Bigg)e^{-\frac{(p+1)x}{a_1\bGSNu}}\ln(1+x)dx.
\end{align}
\end{proposition}
\begin{proof}
By invoking~\eqref{eq:R1R2:cdf},  to derive $\mathcal{R}_{\SUu}^{\Ssss}$, we need to obtain $f_{\gamma_{1,\Ssss}}(\cdot)$.  By using similar steps in deriving~\eqref{eq:Pout:U1:Qos}, we have
\vspace{-0.7em}
\begin{align}\label{eq:pr:P3lfinal}
F_{\gamma_{1,\Ssss}}(x)&=
\sum_{r=0}^{K_1}\binom{\MB\MT}{r}(\Pll)^r (1\!-\Pll)^{\MB\MT-r}
\sum_{i=1}^{\MB}\sum_{\substack{k=1 \\ k \times i = r\\ \text{or}~i=k=r}}^{\MT}
\left( 1- i\sum_{p=0}^{i-1}\!\!
\frac{ (-1)^p\binom{i-1}{p} e^{-\frac{(p+1)x}{a_1\bGSNu}}}
{(p+1)\left(1 \!+\! \frac{(p+1)\bGRNu}{k a_1\bGSNu}x\right)}\!\right).
\end{align}

Accordingly, by taking the first order derivative of $F_{\gamma_{1,\Ssss}}(x)$, the pdf of $\gamma_{1,\Ssss}$ is obtained. Then, by using~\eqref{eq:R1R2:cdf}, the desired result in~\eqref{eq:rate:S3:u1} can be achieved.
\end{proof}

\begin{Remark}
Derivation of the achievable rate of the $\SUuu$ with the QoS provisioning AS scheme is complicated. Moreover, the expressions in~\eqref{eq:R2:S1},~\eqref{eq:R2:S2}, and~\eqref{eq:pr:P3lfinal} are not simple
enough to provide immediate insight, but they are general and fast to be calculated using popular scientific software packages. Hence, we have
resorted to simulations for evaluating the achievable rate of  $\SUuu$ in Section~\ref{sec:Num}.
\end{Remark}

\vspace{-0.0em}
\section{QoS Provisioning AS with Dynamic Antenna Clustering}\label{sec:DQoS}
In the proposed AS schemes, the FD relay's structure is static, i.e., the number of relay's receive and transmit antennas is fixed. Nevertheless, when one can adaptively configure the antennas at relay for transmission and/or reception, the proposed QoS provisioning AS scheme can dynamically choose receive and transmit antennas to ensure the $\SUuu$'s target rate, $\mathcal{R}_2$, while maximizing the achievable rate of $\SUu$. Therefore, the available radio resources will be more efficiently exploited to manage the data rate of different users according to their requirement~\cite{chen2015flexradio}. Antenna clustering has been recently deployed in MIMO communications for simultaneous wireless information and power transfer in reciprocal directions in a point-to-point MIMO-FD system~\cite{Hossain:TCOM:2021}.

In the proposed QoS provisioning AS scheme, it is possible that for a given fixed number of receive and  transmit antennas at $\Rly$, the desired target rate for $\SUuu$ is not realized, i.e., $\mathcal{A}$ is empty. However, for the cases in which $\Rly$ has access to an adaptively configured shared antenna with the ability of operating in either transmission or reception mode, we might configure the number of receive and transmit antennas to meet the rate requirement of the $\SUuu$. On the other hand, it is also possible to realize the $\SUuu$'s target rate with fewer receive and/or transmit antennas. In this case, if $\Rly$  has an adaptively configured shared antenna, the remaining degrees-of-freedom can be utilized further to enhance the achievable rate of $\SUu$. Therefore, there is an optimal antenna configuration at $\Rly$ that improves the performance of the QoS provisioning AS scheme.
\vspace{0em}
\begin{algorithm} [t]
\SetAlgoLined
 \LinesNumbered
 \KwIn{$\hRU$, $\hRUU$, $\hBR$, $\HSI$, M, $\mathcal{R}_2$}
 \KwResult{$\MT^*$, $\MR^*$, $i^*$, $j^*$, $k^*$ }
 Initialization: $\mathcal{A}=\emptyset$, $\MT^\delta=1$, $\MR^\delta=M-1$\\
 \For{$ \ell=1: M $}{
   $\Phi_R=\{m_\ell\}$, $\Phi_T=\{m_1,\ldots,m_{\ell-1},m_{\ell+1},\ldots,m_M\}$ \\
 \While{ $M_T^\delta\geq1$ $\&$ $M_R^\delta\geq1$ }{

 \For{$ i=1: \MB $}{
 \For{$ i=1: M_T^\delta $}{
 \For{$ i=1: M_R^\delta $}{
$\gamma_{2} = \min(\gamma_{12}, \gamma_{\mathtt{R}},\gamma_{\mathtt{RU2}})$\\
 \If{ $\gamma_{2}\geq 2^{\mathcal{R}_2-1}$} {
            $\mathcal{A} = \mathcal{A}\cup\{i,j,k\}$}

  }
  }
  }
  \eIf{$\mathcal{A} \neq \emptyset$}{
   $i^* = \argmax_{1\leq i \leq \MB} \|\HBNu\|^2$\\
   $j^* = \argmax_{ j \in \Phi_R} \|\HBRisj\|^2$\\
   $k^* =\argmin_{k\in\Phi_T}\|\HRNu\|^2$\\
   $\MT^*=|\Phi_R|$, $\MR^*=|\Phi_T|$\\
   Break\;
   }{
   $\Phi_R=\Phi_R \cup m_{\ell+1}$, $\Phi_T=\Phi_T \backslash m_{\ell+1}$\\
   $\MR^\delta=|\Phi_R|$, $\MT^\delta=|\Phi_T|$\;
  }
 }
 \caption{The proposed QoS provisioning AS with dynamic antenna clustering}
 }
\end{algorithm}
Let us denote by $\qC=\{C_{\pi_1},\cdots,C_{\pi_n}\}$ the set of all antenna configurations at the FD relay such that each of the configurations contains at least one transmit-receive antenna pair to simultaneously perform transmission/reception. As an example, let us consider that $\Rly$ is equipped with $M=\MT+\MR=6$ antennas. One possible antenna configuration at $\Rly$ is
\begin{align}
C_{\pi_{\delta}}=\bigg\{ \underbrace{\{m_1,m_5\},}_{\text{transmit antenna indices}}
\underbrace{\{m_2, m_3, m_4, m_6\}}_{\text{receive antenna indices}}\bigg\},
\end{align}
where $1\leq \delta \leq n$ with $n=\sum_{r}^{M-1}\binom{M}{r}$. The optimal set of antenna configuration and then the corresponding selected antenna from receive and transmit antenna subsets must be the ones that produce the best system performance. Specifically, every set of antenna configurations $C_{\pi_\delta},~\forall \delta =1,\ldots,n$, will have the corresponding transmit and receive selected antenna indices.
Accordingly, for the system under consideration, joint antenna clustering and AS problem can be formulated as
 \begin{subequations}\label{eqn:Qos: antconf}
 \begin{align}
   \max_{\qC=\{C_{\pi_1},\ldots,C_{\pi_n}\}} \hspace{1em}&\frac{ \GSNu}{\GRNu+1}\\
    \hspace{-4em}\mbox{s.t.} \hspace{3em}& \mathcal{A}\!=\!
    \Bigg\{\{i,j,k\}\!:\hspace{0em}\min\left(\frac{ a_2\GSNu}{ a_1\GSNu \!+\!\GRNu\!+1},\frac{a_2\GBRij}{ a_1\GBRij \!+\! \GSIkj\!+1},
\GRFu\right)\!>\!2^{R_2}\!-\!1\Bigg\}\!\neq\!\emptyset,\nonumber\\
&\hspace{5em} 1\leq i\leq\MB,~1\leq j\leq\MR^\delta,~1\leq k\leq\MT^\delta\\
    \hspace{1em}& M=\MR^\delta + \MT^\delta,
 \end{align}
  \end{subequations}
where $\MR^\delta$ and $\MT^\delta$ denote the number of receive and transmit antenna indices in $C_{\pi_\delta}$, respectively.

Problem~\eqref{eqn:Qos: antconf} is a combinatorial optimization problem. The complexity of this problem will increase exponentially with the number of antennas (and hence the possible antenna configurations). There is an optimal transmit and receive antenna pair for every possible antenna clustering. The globally optimal solution is the
pair that gives the best performance (i.e., the highest rate at $\SUu$ with intended QoS at $\SUuu$). This performance can only be achieved by finding the best transmit and receive antenna indices for every possible antenna configuration, and therefore, an exhaustive search will be required. In Algorithm 1, we propose a procedure that achieves a balance between complexity and performance. In Algorithm 1, $\Phi_T$ and $\Phi_R$ denote sets containing transmit and receive antennas indices. After antenna clustering, the best transmit and receive antenna is selected such that the received SINR at $\SUu$ is maximized.

Having obtained the optimal antenna cluster at the relay, denoted by $C_{\pi_{\delta^*}}$, the outage probability and the Achievable sum rate of the system under QoS provisioning AS with dynamic antenna clustering can be derived using Propositions~\ref{prop:outage:S3} and~\ref{prop:rate:S3}, respectively, in which $\MR$ and $\MT$ are replaced by $\MR^*$ and $\MT^*$, respectively.


\vspace{-1em}
\section{Implementation and Signal Requirements}~\label{sec:imp}
This section discusses the implementation requirements of the proposed AS schemes. CSI acquisition mechanism at the BS and $\Rly$ to design the AS schemes is described. Then, we elaborate on the computational complexity aspect of the schemes.
\vspace{-1em}
\subsection{Channel State Information Acquisition }
All proposed AS schemes require the availability of CSI at the BS and relay. In general, this CSI acquisition can be achieved either by channel reciprocity in time-division duplexed (TDD) systems or by feedback in frequency-division duplexed (FDD) systems. Since most of the current wireless systems are FDD and the performance of the TDD systems is severely degraded by the phenomena of pilot contamination~\cite{Marzetta:TDD:2011}, we focus on the FDD systems. For the practical implementation of the max-$\SUu$ AS scheme, the BS and $\Rly$   send the pilot signal to $\SUu$. Upon receiving the pilot signal, $\SUu$  feedbacks the  index of the antenna that the BS and $\Rly$ must use in the subsequent information transmission phase. Also, according to the feedback antenna indices, the $\Rly$  decides on the best antenna index to be used on its receive side. In the case of max-$\SUuu$ AS schemes, the BS transmits a pilot signal to $\SUu$ and $\Rly$, and $\Rly$ transmits the pilot signal to $\SUu$ and $\SUuu$. Upon reception of the pilot, $\SUuu$ can next feedback the antenna index that the $\Rly$ must use at the next transmission phase, while the $\SUu$ feedbacks the antenna index used at the BS. Then, the relay can decide on the best antenna index for its receive side. For the QoS provisioning AS scheme, the same process as in max-$\SUuu$ AS schemes can be performed to first create the antenna set $\mathcal{A}$. Then, based on the available feedback from $\SUu$ and $\SUuu$ at the BS and $\Rly$, the best antenna indices can be selected.

In practice, we notice that the perfect CSI is challenging to obtain due to the impact of channel estimation error and limited feedback. Moreover, the CSI used for  TAS is usually outdated due to feedback delay. Moreover, the imperfect/outdated CSI can deteriorate the SIC process of the NOMA system~\cite{Ding:Survay}. Our results provide insights on the efficient design of TAS in cooperative NOMA systems and open several exciting avenues for future research, including analyzing the impact of imperfect/outdated CSI on TAS problems in cooperative FD NOMA systems.

\vspace{-1em}
\subsection{Complexity Analysis}
For the proposed AS schemes, the BS and $\Rly$ require each coherent interval to perform channel estimation and AS to enable efficient data transmission between the BS and NOMA users. Higher complexity results in longer processing time, reducing  the data transmission interval. Therefore, the complexity analysis in this subsection is essential to quantify the efficiency of the proposed AS schemes for practical scenarios. We follow the complexity analysis in~\cite{Gershman:JSP}.

For the max-$\SUu$ AS, the relay transmit antenna index, $k^*$, is determined by first computing $|\HRNu|$ for all $k=1,\ldots,\MT$ and then sort the $\MT$ norms, which requires $\mathcal{O}(\MT+\MT\log\MT)$ complexity. Similarly, the required computational complexity to find the BS transmit antenna index $i^*$ and relay receive antenna index $j^*$  are $\mathcal{O}\left(\MB+\MB\log\MB\right)$ and $\mathcal{O}\left(\MR+\MR\log\MR\right)$, respectively.  Recalling that max-$\SUuu$ AS first determines $k^*$ such that $\GRFu$ is maximized. It then selects the weakest SI channel $\GSIkj$  for the given $k^*$ and finally finds $i^*$ such that for the given $j^*$, $\GBRij$ is maximized. Hence, the total computational complexity of max-$\SUuu$ is the same as the max-$\SUu$ AS scheme.

With the QoS provisioning AS scheme, after construction set $\mathcal{A}$, the optimum subset of $\mathcal{A}$, i.e.,  $\{i^*, j^*, k^*\}$, is determined such that \emph{e2e} SINR at $\SUu$ is maximized. Since $\MB\MR\MT$ possible subsets of $\mathcal{A}$ must be checked to find out the potential subsets ensuring target rate $\mathcal{R}_2$ for $\SUuu$, the rough complexity of the QoS AS scheme is determined as  $\mathcal{O}\big( \MB\MR\MT((\MB+ \MR+ \MT)+(\MB\log\MB+\MR\log\MR+\MT\log\MT))\big)$. When dynamic antenna clustering is considered at $\Rly$, the number of candidate subsets in $\mathcal{A}$ is $\sum_{r=1}^{M-1}\binom{M}{r}$. Hence, the complexity of the QoS AS scheme is $\mathcal{O}\big( \sum_{r=1}^{M-1}\binom{M}{r}((\MB+ \MR+ \MT)+(\MB\log\MB+\MR\log\MR+\MT\log\MT))\big)$,  which is fairly high compared to the fixed antenna setup at $\Rly$, especially for the large value of $M=\MR+\MT$.

Moreover, the computational complexity of all schemes depends on the antenna numbers at both BS and $\Rly$. Therefore, the QoS provisioning AS scheme with dynamic antenna clustering requires the highest computational complexity. On the other hand, the computational complexity of the max-$\SUu$ and max-$\SUuu$ AS schemes are the same and relatively low, albeit at the cost of losing performance (c.f. Fig.~\ref{fig:Rd}).

\vspace{-1em}
\section{Numerical Results and Discussion} \label{sec:Num}
This section presents numerical results to quantify the performance gains when the considered FD cooperative NOMA system adopts the max-$\SUu$, max-$\SUuu$, and QoS provisioning AS schemes. Achievable rates of the QoS provisioning AS scheme with configurable shared antenna at $\Rly$ are presented to demonstrate the potential impact of dynamic antenna clustering on the AS.
Unless otherwise stated, the value of network parameters are: $a_1=0.1$, $a_2=0.9$, $\sigma^2_{\mathtt{RU1}}=\sigma^2_{\mathtt{RU2}} =0.5$, $\sigma^2_{\mathtt{SU1}}=\sigma^2_{\mathtt{SR}}=1$, and $\Sap = 0.3$. In order to verify the advantage of the proposed AS schemes, three benchmark schemes are considered with fixed antenna setup at $\Rly$ \emph{i) {Optimum AS scheme}}: It performs an exhaustive search of all possible combinations to determine the antenna subset in order to maximize the ergodic sum rate, \emph{ii) {Optimum-$\SUuu$ AS scheme} }: This scheme aims to maximize the \emph{e2e} SINR at $\SUuu$ in an optimal sense and performs an exhaustive search of all possible combinations to determine the optimum antenna subset, and \emph{iii) {Random AS scheme}}: It performs random AS at the BS and relay input/output.

\begin{figure}[tbp]
\begin{subfigure}[a]{0.5\textwidth}
\includegraphics[width=90mm, height=67mm]{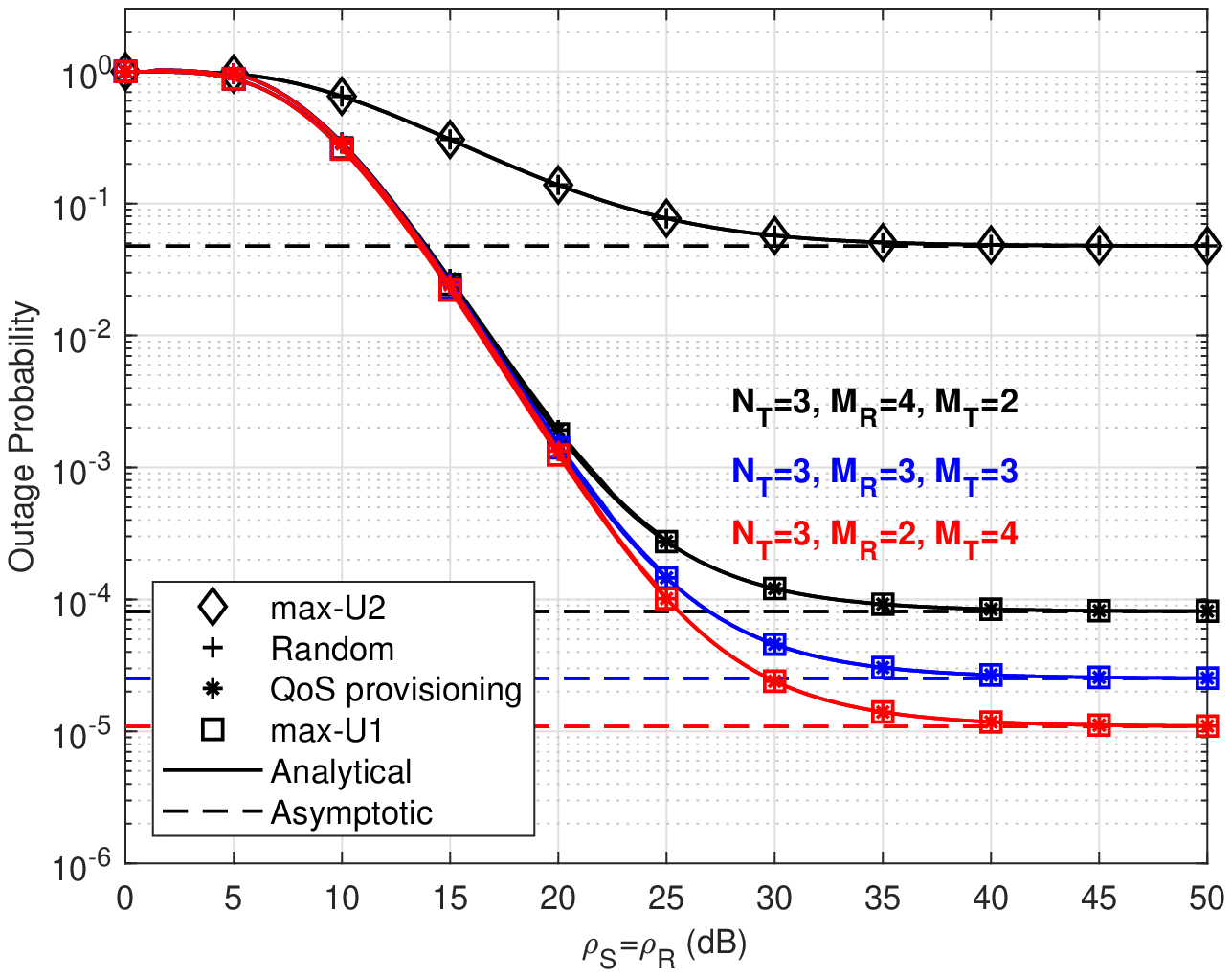}
\caption{ $\SUu$.}
\label{fig:Poutu1}
\end{subfigure}
\begin{subfigure}[a]{0.5\textwidth}
\includegraphics[width=90mm, height=67mm]{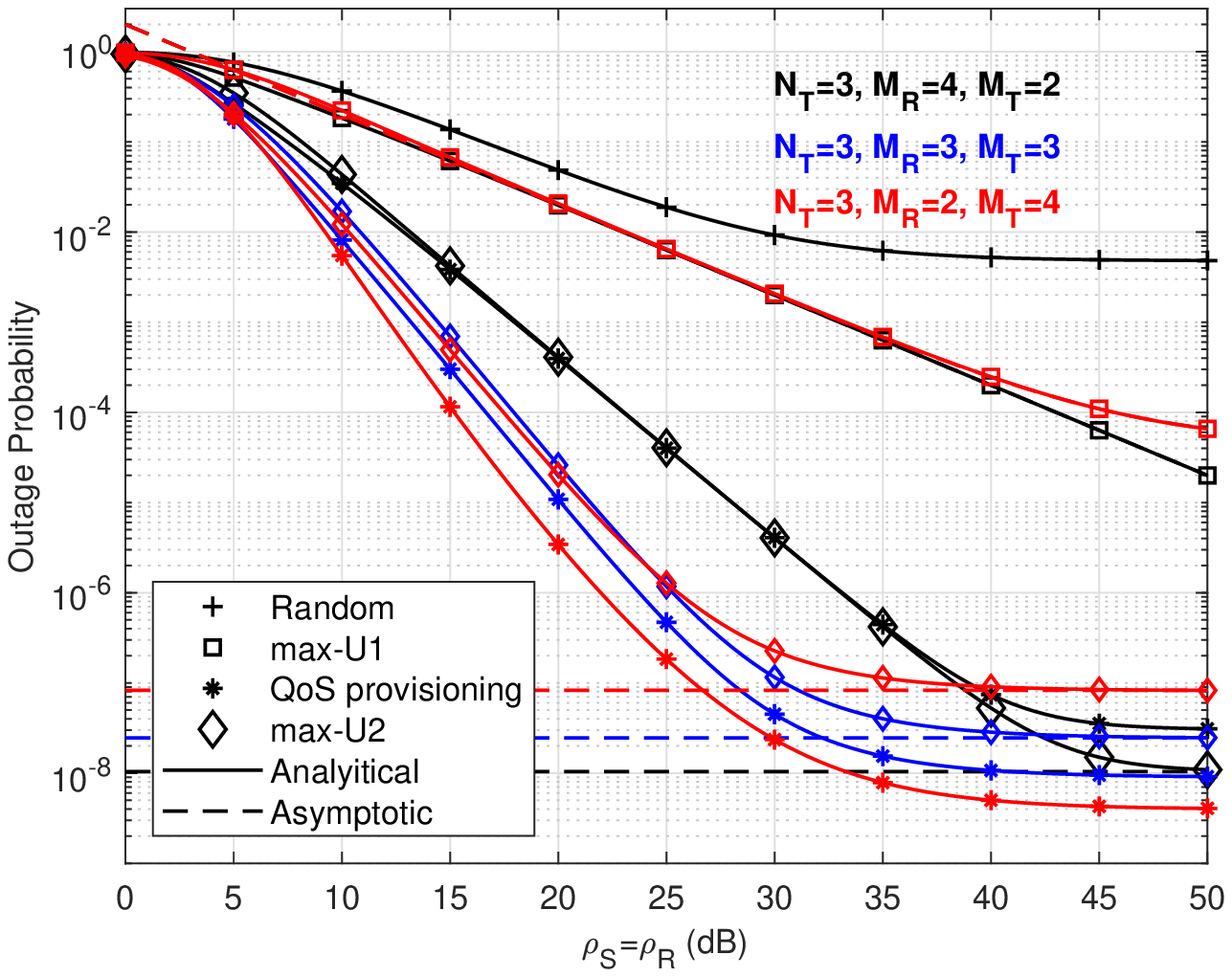}
\caption{$\SUuu$.}
\label{fig:Poutu2}
\end{subfigure}
\caption{Outage probability of $\SUu$ and $\SUuu$ for different antenna configurations ($k_1=0.1$, $\mathcal{R}_1=\mathcal{R}_2=1$ (bit/sec/Hz)).}\label{fig:Pout}
\vspace{-2em}
\end{figure}
\vspace{-0.5em}
\subsection{Outage Probability}
Fig.~\ref{fig:Pout} shows the outage probability of the near user $\SUu$ and far user $\SUuu$ versus $\rho_S=\rho_P$  for the proposed AS schemes with different antenna configurations at relay, where the analytical results are based on Propositions~\ref{outageU1:S1S2},~\ref{outageU2:S1S2}, and~\ref{prop:outage:S3} and the asymptotic results are
based on Corollary~\ref{Corr:S1S2}. We observe that the analytical results (solid lines) tightly match
simulation results (marker lines) and that asymptotic curves (dashed lines) tightly converge
to the exact ones at the high-SNR regime. Comparing the QoS provisioning and  max-$\SUu$ AS schemes for $\SUu$ in Fig.~\ref{fig:Poutu1}, we see that their outage performances are  the same  in almost all transmit power regimes. Nevertheless, for minimal values of $P$, there is a negligible performance difference between the QoS provisioning  and  max-$\SUu$ schemes. The intuitive reason is that the former must simultaneously ensure $\SUuu$'s targeted data rate and serve $\SUu$ with a rate as large as possible. Therefore, its outage performance depends on the  set $\mathcal{A}$, given  in~\eqref{eq:TwoStage:far}, and is more efficient for  scenarios with larger $P$ which provide  larger $|\mathcal{A}|$, i.e., more AS subsets choices. For $\MR +\MT=6$, the additional transmit antenna, e.g., $\MT=4$ and $\MR=2$, could increase the degree-of-freedom (in terms of AS) to maximize the \emph{e2e} SINR at the $\SUu$ and enhance the outage performance. Additionally, we observe that the  max-$\SUu$ and random schemes exhibit the same outage probability of the $\SUu$ for all antenna configurations. This is because max-$\SUu$ uses all degrees of spatial freedom  to maximize the \emph{e2e} at the $\SUuu$. Hence, there is no spare capacity to maximize the \emph{e2e} SINR at $\SUu$.

\begin{figure}[tbp]
\begin{subfigure}[a]{0.5\textwidth}
\includegraphics[width=90mm, height=67mm]{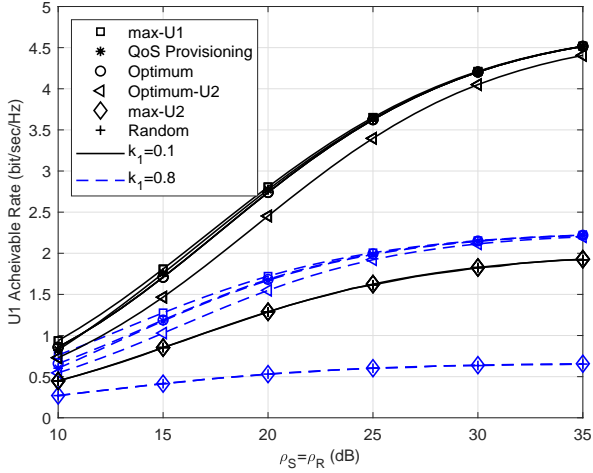}
\caption{$\SUu$.}
\label{fig:Ru1}
\end{subfigure}
\begin{subfigure}[a]{0.5\textwidth}
\includegraphics[width=90mm, height=67mm]{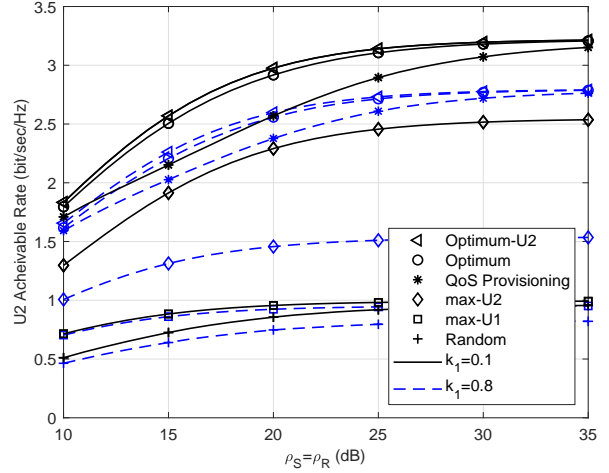}
\caption{$\SUuu$.}
\label{fig:Ru2}
\end{subfigure}
\caption{The achievable rate  for different inter-user interference strengths ($\MB=\MR=\MT=4$, $\mathcal{R}_2=1.5$ (bit/sec/Hz)).}\label{fig:Ru12}
\vspace{-2.1em}
\end{figure}

In Fig.~\ref{fig:Poutu2} we compare the outage probability of the proposed AS schemes with different antenna configurations for $\SUuu$. Fig.~\ref{fig:Poutu2} shows the superiority of the QoS provisioning AS scheme over other AS schemas, which improves with the increasing transmission power. Nevertheless, when SNR is high,  $\MR=4$, and  $\MT=2$, max-$\SUuu$ outperforms other schemes. This out-performance happens because decreasing $\MT$ decreases the $|\mathcal{A}|$ which reduces the degree-of-freedom available to maximize the second-hop SINR at $\SUuu$. Moreover, the QoS provisioning AS scheme is more favorable than other AS schemes for the scenarios with a smaller number of receive antennas at $\Rly$. For the max-$\SUuu$, we see that in the low-to-medium SNR regime, an additional transmit antenna at $\Rly$ could increase the SINR of the second hop and enhance the outage performance of $\SUuu$.   However, at high SNR, max-$\SUuu$ is more efficient with more receive antennas at $\Rly$.

\begin{figure}[tbp]
\begin{subfigure}[a]{0.5\textwidth}
\includegraphics[width=90mm, height=67mm]{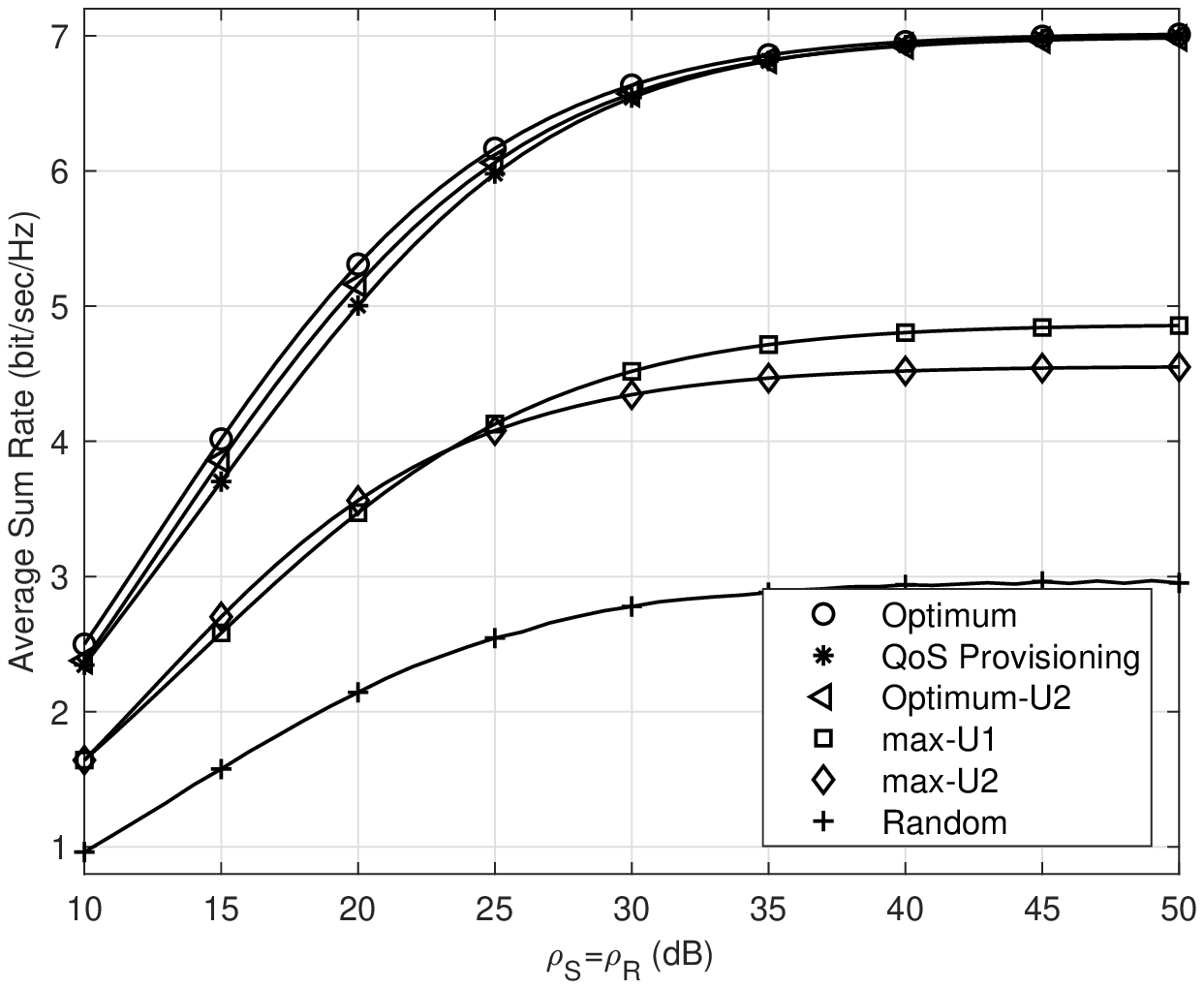}
\caption{$\MB=4$, $\MR=6$, $\MT=2$.}
\label{fig:Rsum1}
\end{subfigure}
\begin{subfigure}[a]{0.5\textwidth}
\includegraphics[width=90mm, height=67mm]{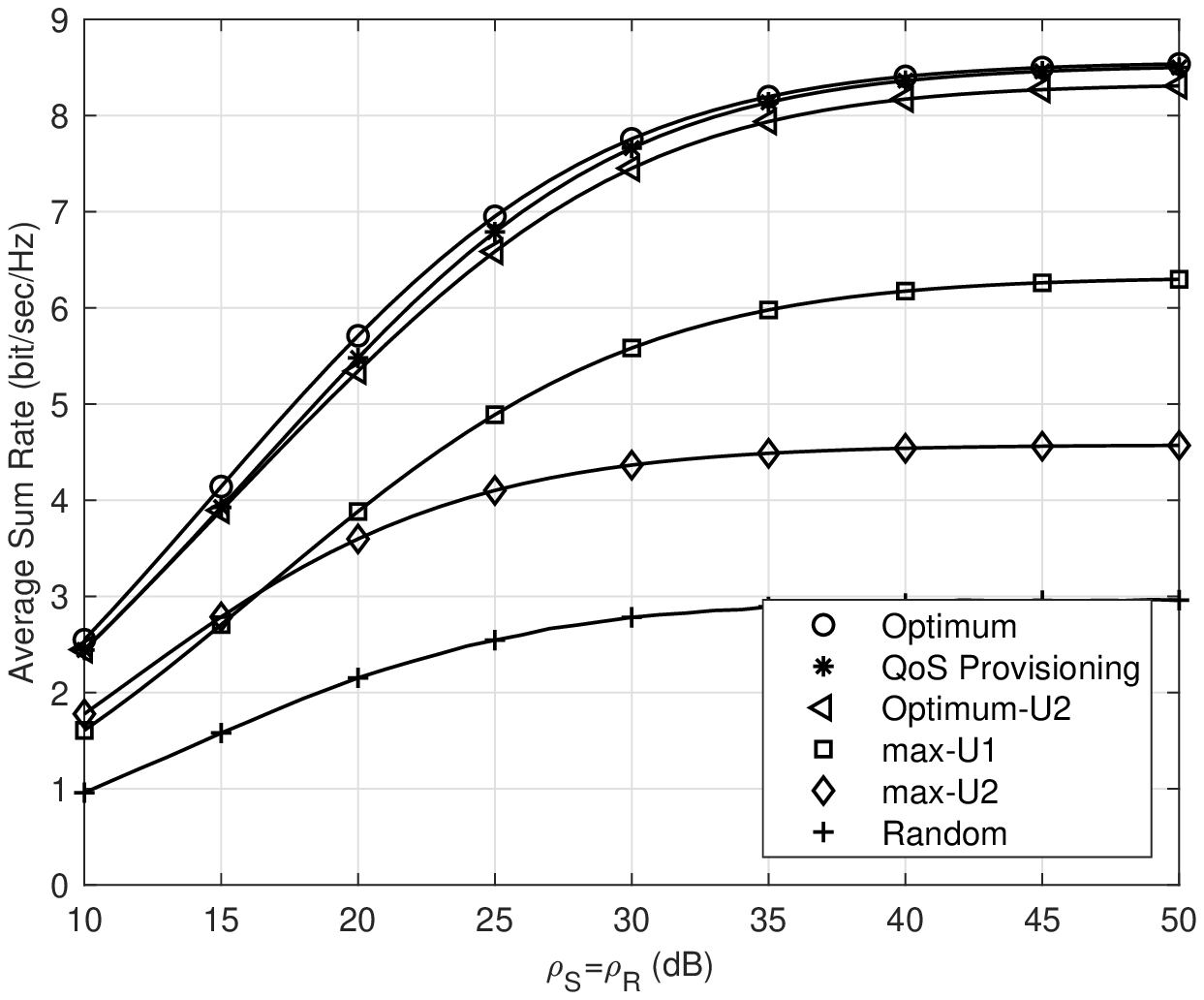}
\caption{$\MB=4$, $\MR=2$, $\MT=6$.}
\label{fig:Rsum2}
\end{subfigure}
\caption{The achievable rate  for different antenna configurations at $\Rly$ ($k_1=0.1$, $\mathcal{R}_2=1.5$ (bit/sec/Hz)).}
\label{fig:Rsum}
\vspace{-2.1em}
\end{figure}
\vspace{-0.7em}
\subsection{Achievable Rate}
Fig.~\ref{fig:Ru12} compares the achievable rate at $\SUu$ and $\SUuu$  due to the optimum, optimum-$\SUuu$,  QoS provisioning, max-$\SUu$, max$\SUuu$, and random AS schemes. Fig.~\ref{fig:Ru1} depicts the superiority of the optimum,  QoS provisioning, and max-$\SUu$ schemes over optimum-$\SUuu$, max-$\SUuu$, and random AS schemes for $\SUu$. In addition, we observe that as $k_1$ decreases, all AS schemes achieve better achievable rates since lower $k_1$ causes weaker interference between the relay and $\SUu$. Notably, the max-$\SUu$ and QoS provisioning schemes exhibit the best achievable rate performance. The superior performance of these schemes is more pronounced in the low-to-medium SNR regime. Fig.~\ref{fig:Ru2} illustrates that optimum-$\SUuu$, optimum, and QoS provisioning AS schemes outperform other AS schemes in  all transmit power regimes. Nevertheless,  the QoS provisioning scheme achieves a superior achievable rate performance at $\SUu$ compared to the optimum-$\SUuu$ and optimum schemes. We observe that the $\SUuu$'s targeted data rate $\mathcal{R}_2=1.5$ bit/sec/Hz is not realized by max-$\SUu$, max-$\SUuu$, and random AS schemes.

Fig.~\ref{fig:Rsum} compares the sum achievable rate of the proposed AS schemes versus $\rho_S=\rho_P$ for $M=\MR+\MT$ antennas at the relay. Two cases are considered where  ($\MR=6$, $\MT=2$) in Fig.~\ref{fig:Rsum1} and ($\MR=2$, $\MT=6$)  in Fig.~\ref{fig:Rsum2}. We see that the QoS provisioning AS scheme exhibits a near-optimum performance. More specifically, for the combination ($\MR=2$, $\MT=6$), the QoS provisioning scheme outperforms the optimum-$\SUuu$ scheme. However, the outage performance of max-$\SUuu$ and random schemes is not sensitive to the antenna configuration at the relay. This is because max-$\SUuu$ uses all spatial  degrees of freedom to only maximize the \emph{e2e} at the $\SUuu$. In contrast, the considered criterion  for AS selection in QoS provisioning essentially distributes the spatial degrees of freedom among $\SUu$ and $\SUuu$, providing user fairness. An interesting observation is that the sum achievable rate of the system significantly depends on the relay's antenna configuration. In particular, with ($\MR=6$, $\MT=2$) antenna configuration at the relay,  the achievable rates of the QoS provisioning and  max-$\SUu$ improve up to $23\%$ and $31\%$, respectively, compared to ($\MR=2$, $\MT=6$) configuration. Finally, Figs.~\ref{fig:Rsum} depict that the achievable rate of  $\SUu$ and $\SUuu$  shows a floor at high power values for all AS schemes. This is expected because the inter-user interference and self-interference caused by the relay's FD operation will be maximal with high relay transmit power, which reduces the achievable rate.

\begin{figure}[tbp]
\begin{subfigure}[a]{0.32\textwidth}
\includegraphics[width=59mm, height=47mm]{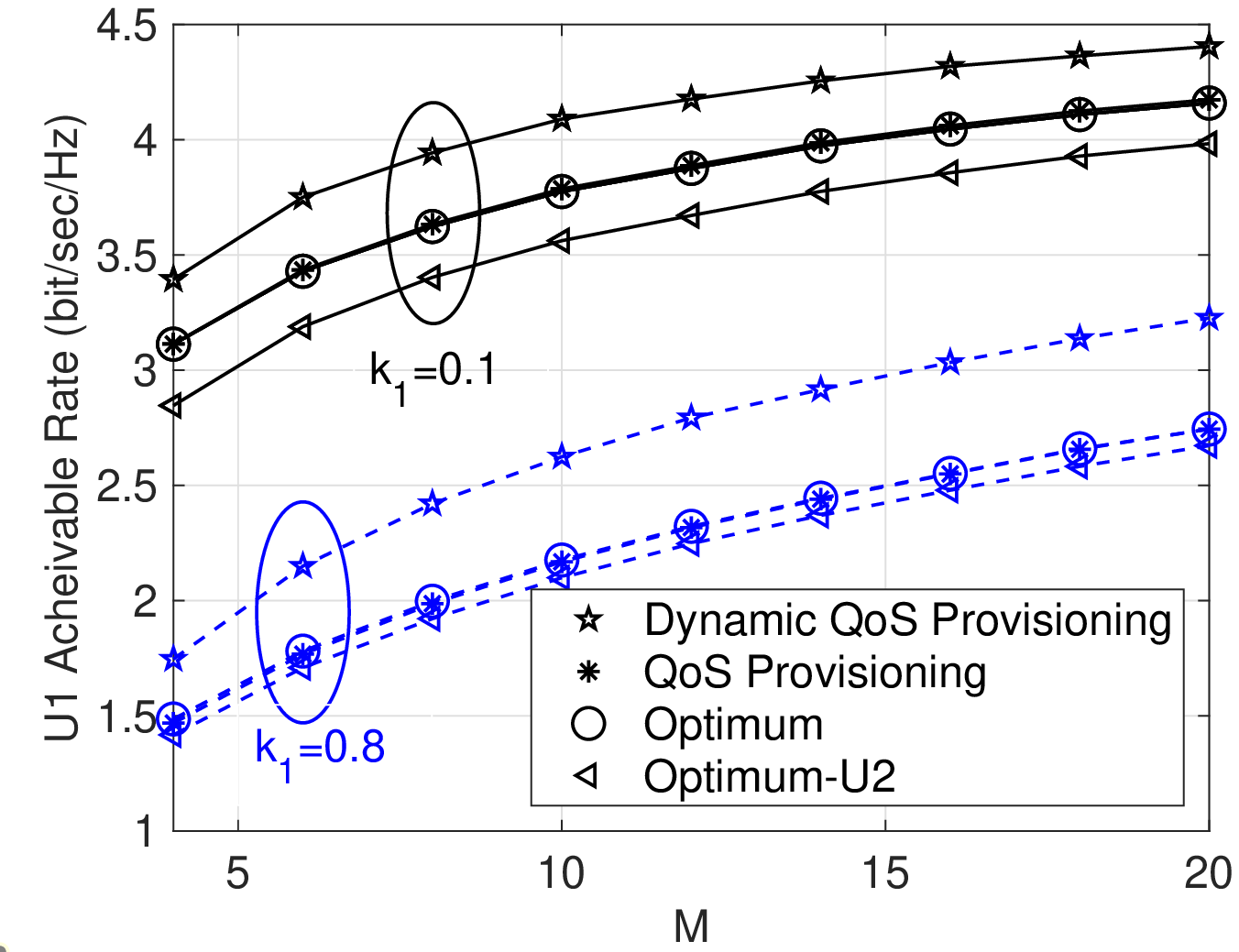}
\caption{$\SUu$.}
\label{fig:RU1d}
\end{subfigure}
\begin{subfigure}[a]{0.32\textwidth}
\includegraphics[width=59mm, height=47mm]{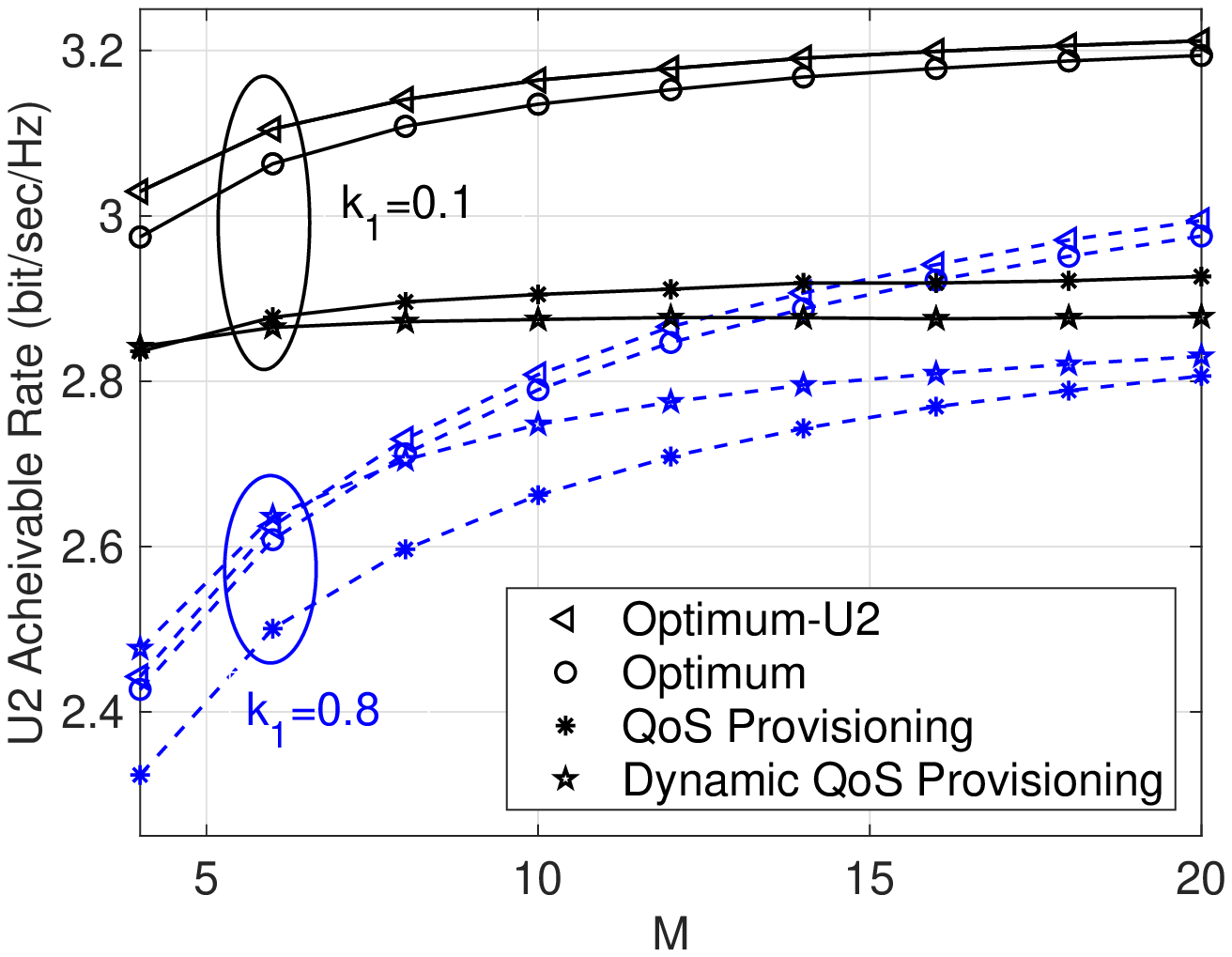}
\caption{$\SUuu$.}
\label{fig:RU2d}
\end{subfigure}
\begin{subfigure}[a]{0.32\textwidth}
\includegraphics[width=60mm, height=47mm]{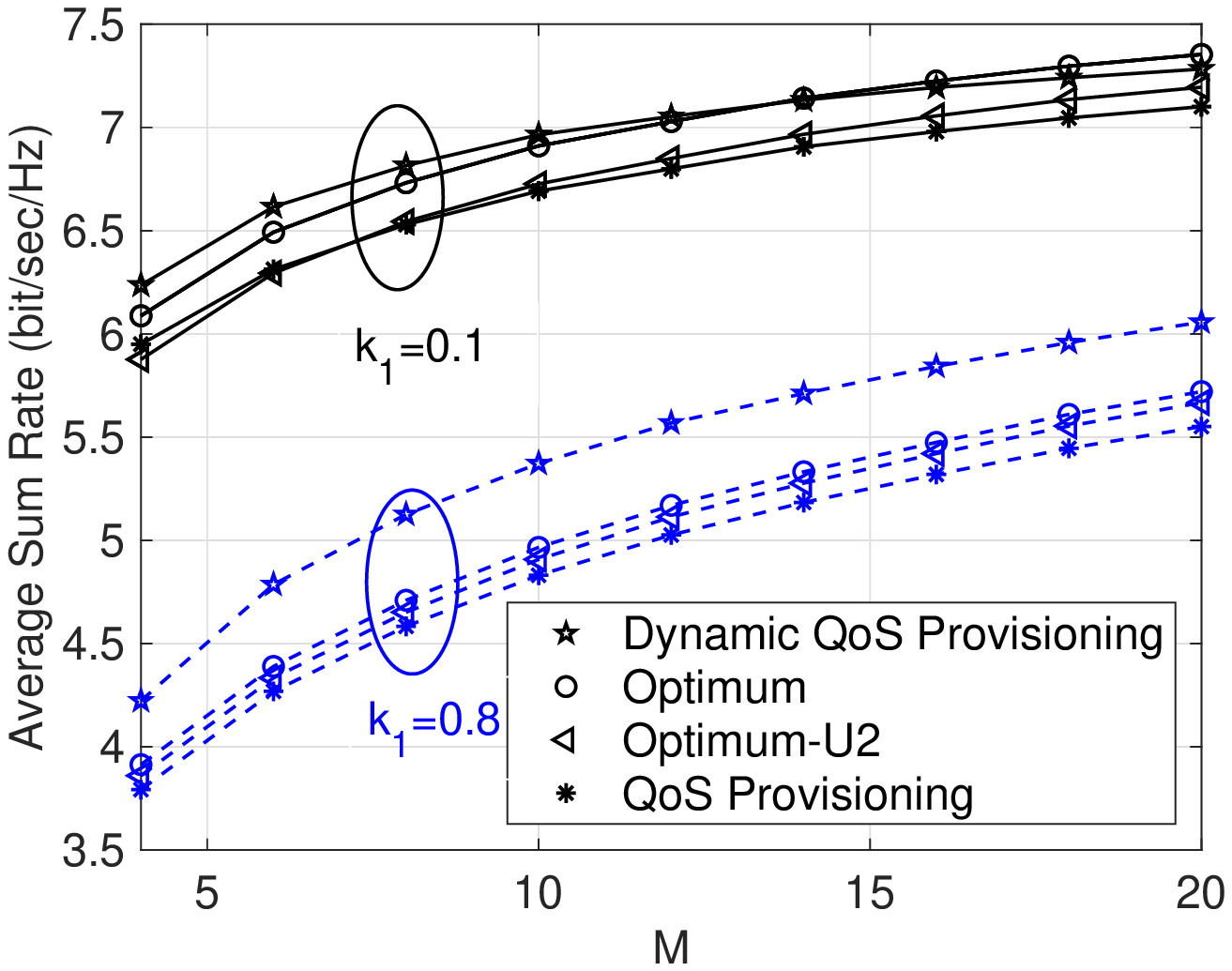}
\caption{Sum rate.}
\label{fig:RSUMd}
\end{subfigure}
\vspace{-0.5em}
\caption{The achievable rate versus different antenna numbers at $\Rly$ ($\MB=4$, $\MR=\MT=M/2$, $\mathcal{R}_2=1.5$ (bit/sec/Hz), $\rho_S=\rho_R=25$ dB).}
\label{fig:Rd}
\vspace{-2em}
\end{figure}
\vspace{-1em}
\subsection{Dynamic Antenna Clustering}

We evaluate the effect of dynamic antenna clustering at the relay on the achievable rate of $\SUu$, $\SUuu$, and sum rate.
Fig.~\ref{fig:Rd} shows the achievable rate of the proposed QoS provisioning AS scheme with dynamic antenna clustering and static antenna setup at the relay and for two different inter-user  interference scenarios: 1) weak interference ($k_1=0.1$), and 2) substantial interference ($k_1=0.8$). We have further included the optimum and optimum-$\SUuu$ AS scheme for comparison. In this example, we assume the number of relay's receive and transmit antenna to be $\MR=\MT=M/2$ for  QoS provisioning, optimum, and optimum-$\SUuu$ with static antenna configuration at the relay. In the case of QoS provisioning AS with dynamic antenna clustering, we set $\MR=\MR^*$ and $\MT=\MT^*$, determined by Algorithm 1. From Fig.~\ref{fig:RU1d}, it is evident that the QoS provisioning AS scheme with dynamic antenna clustering provides a more achievable rate than all other schemes for different interference levels. This can be explained as follows. According to the AS scheme criterion in~\eqref{eq:TwoStage:far},  this scheme instantly, after ensuring the target rate $\mathcal{R}_2$ for $\SUuu$, focuses on maximizing $\SUu$'s achievable rate. Therefore, this scheme provides more degrees of freedom in choosing the links between $\Rly$ and $\SUuu$, resulting in less severe interference at $\SUuu$. We can observe this trend from Fig.~\ref{fig:RU2d}, where the achievable rate of $\SUuu$ achieved by the QoS provisioning AS scheme (both static and dynamic) remains relatively constant by increasing $M$. Finally, we can see that the average sum rate of the QoS provisioning AS scheme with dynamic antenna clustering outperforms all schemes, especially in strong inter-user interference scenarios.

\vspace{-0.7em}
\section{Conclusion} \label{sec:Conc}
This paper has studied the AS problem for an FD cooperative NOMA system. Two low complexity AS schemes, namely max-$\SUu$ AS scheme and  max-$\SUuu$ AS scheme, were proposed, aiming to maximize the \emph{e2e} SINR at the near and far user, respectively. In order to satisfy the rate requirement of $\SUuu$ and maximize the achievable rate of $\SUuu$ at the same time,  we further proposed a QoS provisioning AS scheme with static and dynamic antenna clustering at the relay. Dynamic antenna clustering is applicable for the case $\Rly$ has access to an adaptively configured shared antenna with the ability to operate in either transmission or reception mode. Thus, we proposed dynamic antenna clustering and AS to configure the number of receive and transmit antennas at $\Rly$ to meet the rate requirement of the $\SUuu$ and maximize the achievable rate in $\SUuu$, leading to fairness in respect to near and far user's target rates. We have characterized the performance of proposed AS schemes in terms of the outage probability and achievable rate. Our results revealed that the QoS provisioning scheme with dynamic antenna clustering achieves up to $12\%$ and $10\%$ average sum rate gains compared with the QoS provisioning and optimum AS scheme with static antenna setup, respectively, in the intense interference scenarios.

Extending these results to the millimeter-wave massive MIMO systems would be interesting, including joint AS and hybrid beamforming design problems~\cite{Chintha:JSTSP:2018}. Moreover, machine learning techniques may also be leveraged to solve these problems.

\appendices
\section{Proof of Proposition~\ref{outageU1:S1S2}}
\label{Appendix:outageU1:S1S2}
By invoking~\eqref{eq:outnear}, we now proceed to derive $F_{\gamma_{1,\mathsf{AS}}}(\cdot)$ for max-$\SUu$ scheme. By inspecting~\eqref{eq:ASc near}, the ratio $\frac{\GSNu}{\GRNu+1}$ is maximized when the strongest BS-$\SUu$ channel and the weakest $\Rly$-$\SUu$ channel are selected. Therefore, the cdf of $\gamma_{1,\Ss}$ can be
evaluated as
\begin{align}~\label{eq:Fgam1S1}
F_{\gamma_{1,\Ss}}(x)=\int_{0}^{\infty} F_{A_1}\left((y+1)x\right)f_{B_1}(y)dy,
 \end{align}
 where $A_1$ is a RV defined as the maximum out of $\MB$ exponentially distributed independent RVs with parameter $\bGSNu$, while $B_1$ is the minimum out of $\MT$ exponentially distributed independent RVs with parameter $\bGRNu$. Substituting the cdf and the pdf of $A$ and $B$ into~\eqref{eq:Fgam1S1}, we arrive at
\begin{align}\label{eq:cdf:gam1:S1:int}
F_{\gamma_{1,\Ss}}(x) &= 1- \frac{\MB\MT}{\bGRNu}\sum_{p=0}^{\MB-1}
\frac{ (-1)^p\binom{\MB-1}{p} }
{(p+1)}\int_{0}^{\infty}e^{-\frac{(p+1)(y+1)x}{\bGSNu}}e^{-\frac{\MT }{\bGRNu}y}dy,
\end{align}
where we have used the fact $F_{A_1}\left(x\right)=\left(1-e^{-\frac{x}{\bGSNu}}\right)^{\MB}=1-\MB\sum_{p=0}^{\MB-1}\frac{(-1)^p\binom{\MB-1}{p}}{p+1} e^{-\frac{(p+1)x}{\bGSNu}}$.
Applying the integral identity~\cite[Eq. (3.310)]{Integral:Series:Ryzhik:1992}, the integral in~\eqref{eq:cdf:gam1:S1:int}
can be solved as
\vspace{-.1em}
\begin{align}\label{eq:cdf:gam1:S1}
F_{\gamma_{1,\Ss}}(x) &= 1- \MB\sum_{p=0}^{\MB-1}\!\!
\frac{ (-1)^p\binom{\MB-1}{p} e^{-\frac{(p+1)x}{\bGSNu}}}
{(p+1)\left(1 + \frac{(p+1)\bGRNu}{\MT \bGSNu}x\right)}.
\end{align}

Now we turn our attention to derive $F_{\gamma_{1,\Sss}}(\cdot)$. According to~\eqref{eq:ASc far} a single transmit antenna at $\Rly$ is selected  such that the received SNR at $\SUu$ is maximized. Moreover, a single transmit antenna at the BS is selected such that the SINR at BS-$\Rly$ is maximized. Therefore, $\GSNu$ and $\GRNu$  are exponential RVs with parameter $\bGSNu$ and $\bGRNu$, respectively. Therefore, by using the order statistics, we obtain
\vspace{-0.2em}
\begin{align}\label{eq:cdf:gam1:S2}
F_{\gamma_{1,\Sss}}(x) &= 1- \frac{e^{-\frac{x}{a_1\bGSNu}}}{1+\frac{\bGRNu}{a_1\bGSNu}x},
\end{align}

To this end, the desired result is obtained by evaluating~\eqref{eq:cdf:gam1:S1} and~\eqref{eq:cdf:gam1:S2} at $\zeta$.

\vspace{-0.6em}
\section{Proof of Proposition~\ref{outageU2:S1S2}}
\label{Appendix:outageU2:S1S2}
We first derive the outage probability of  $\SUuu$ with max-$\SUu$ antenna selection scheme. Based on~\eqref{eq:ASc near}, for the selected transmit antennas at the BS and $\Rly$, the ratio $\frac{a_2\GBRiij}{ a_1\GBRiij + \GSIkkj+1}$ can be maximized when the strongest BS-$\Rly$ channel and weakest SI channel are selected.\footnote{We notice that according to~\eqref{eq:ASc near}, transmit antennas at the BS and $\Rly$ are fixed as they have already selected to maximize the \emph{e2e} SINR  at $\SUu$.} However, theses two channels are coupled with each other through the selected antenna at the $\Rly$ input. Therefore, it is difficult, if not impossible, to find the cdf of $\frac{a_2\GBRiij}{ a_1\GBRiij + \GSIkkj+1}$. Alternatively, we propose to select the receive antenna at $\Rly$ such that $\GBRiij$ is maximized\footnote{In Section~\ref{sec:Num}, it is shown that this approximation  is a reasonable across the entire SNR range (cf. Optimum-$\SUuu$ AS scheme).}. Therefore,  we have
\vspace{-0.4em}
\begin{align*}
\Prob\left(\GMRS>x\right)=1-\int_{0}^{\infty} F_{A_2}\left((y+1)\cnox\right)f_{B_2}(y)dy,
\end{align*}
for $x<\frac{a_2}{a_1}$, where $A_2$ is a RV defined as the largest out of $\MR$ exponentially distributed independent RVs with parameter $\bGBR$, and since SI link is ignored, $B_2$ is an exponentially distributed RV with parameter $\bGSI$. Substituting the required cdf and the pdf and simplifying yields
\vspace{-0.5em}
\begin{align}\label{eq:cdf:GMRS}
\Prob\left(\GMRS\!>\!x\right)& = \MR\sum_{q=0}^{\MR-1}\frac{(-1)^q\binom{\MR-1}{q}e^{-\frac{(q+1)\cnox}{\bGBR}}}
{(q+1)\left(1+\frac{\bGSI}{\bGBR}(q+1)\cnox\right)}.
\end{align}

Moreover, since the $\Rly$-$\SUuu$ link is ignored, $\GRFus$ is an exponentially distributed RV with parameter $\bGRFu$ and thus we have
\vspace{-1.2em}
\begin{align}\label{eq:cdf:GRFus}
\Prob(\GRFus>x) = e^{-\frac{x}{\bGRFu}}.
 \end{align}
To this end by substituting~\eqref{eq:cdf:GMRS} and~\eqref{eq:cdf:GRFus} into~\eqref{eq:outfar} we arrive at~\eqref{eq:outnuS2}.

Now we turn our attention to derive $\PoutfuSS$. According to~\eqref{eq:ASc far} the \emph{e2e} SINR at $\SUuu$ is maximized when each term inside in the minimum function is maximized. Therefore, a transmit antenna at $\Rly$ is selected  such that $\GRFu$ is maximized, i.e., $\GRFu$ is the largest of  $\MT$ exponential RVs with parameter $\bGRFu$. Therefore, we get
\vspace{-1.1em}
\begin{align}\label{eq:cdf:GRFuss}
\Prob(\GRFuss>x) = \MT\!\sum_{q=0}^{\MT-1}
\frac{ (-1)^q\binom{\MT-1}{q} e^{-\frac{(q+1)\theta_2}{\bGRFu}}}
{(q+1)}.
 \end{align}

Moreover, since SI is the main source of performance degradation in FD mode, for a particular transmit antenna at $\Rly$, the best receive antenna at the $\Rly$ is selected such that the SI strength is minimized~\cite{Suraweera:TWC:2014}. Hence, $B_3\triangleq\GSIkkj$ is the minimum of $\MR$ exponential RVs with parameter $\bGSI$. Finally, for given $k^*$ and $j^{*}$, a best transmit antenna at the BS is selected such that $\gamma_{\mathtt{R}}$ is maximized. Hence, a single transmit antenna at the BS is selected such that the SINR at BS-$\Rly$ is maximized for the $j^*$-th receive antenna at $\Rly$, i.e., $A_3\triangleq\gamma_{\mathtt{SR}}^{i,j^*}$ is a RV defined as the maximum out of $\MB$ exponentially distributed independent RVs with parameter $\bGBR$. Therefore, we have
\begin{align}\label{eq:cdf:GMRSs}
\Prob\left(\GMRSS>x\right)& =1-\int_{0}^{\infty} F_{A_3}\left((y+1)\cnox\right)f_{B_3}(y)dy.
\end{align}

By using the required cdfs and  pdfs and simplifying yields
\begin{align}\label{eq:cdf:GMRSs}
\Prob\left(\GMRSS\!>\!x\right)&\! =\MB\sum_{p=0}^{\MB-1}\frac{(-1)^p\binom{\MB-1}{p}e^{-\frac{(p+1)\cno}{\bGBR}}}
{(p\!+\!1)\left(1\!+\!\frac{\bGSI}{\MR\bGBR}(p+1)\cno\right)}.
\end{align}

To this end by substituting~\eqref{eq:cdf:GRFuss} and~\eqref{eq:cdf:GMRSs} into~\eqref{eq:outfar} we arrive at~\eqref{eq:outfuS2}.

\section{Proof of Proposition~\ref{prop:Ru1u2:S1}}
\label{proof:prop:Ru1u2:S1}
It is notable that for a nonnegative RV $X$  the achievable rate in~\eqref{eq:R1R2:cdf}  can be expressed as~\cite{Mohammadi:TCOM:2015}
\begin{align}\label{eq:rate and cdf}
{R}_{\SUui}^{\mathtt{AS}}=
\frac{1}{\mathrm{ln} 2}
\int_{0}^{\infty}\frac{1-F_{\gamma_{u,\mathtt{AS}}}(x)}{1+x}dx.
\end{align}
Based on~\eqref{eq:rate and cdf}, the ergodic achievable rate can be derived using the cdfs given in~\eqref{eq:SINR at UE1} and~\eqref{eq:e-2-e far}. First, we calculate ${R}_{\SUu}^{\Ss}$. By substituting~\eqref{eq:cdf:gam1:S1} into~\eqref{eq:rate and cdf}, we have
\begin{align}\label{eq:R_SU1}
{R}_{\SUu}^{\Ss}=
\frac{\MB}{\mathrm{ln} 2}
\sum_{p=0}^{\MB-1}
\int_{0}^{\infty}
\frac{ (-1)^p\binom{\MB-1}{p} e^{-\frac{(p+1)x}{a_1\bGSNu}}}
{(p+1)\left(1 \!+\! \frac{(p+1)\bGRNu x}{\MT a_1\bGSNu}\right)(1+x)} dx.
\end{align}

To this end, using the integration identity in~\cite[Eq. (3.352.4)]{Integral:Series:Ryzhik:1992}, and after some algebraic manipulations, \eqref{eq:R1:S1} can be obtained.

In order to derive $\mathcal{R}_{\SUuu}^{\Ss}$, we need the cdf of $\gamma_{2,\Ss}$, which can be written as
\begin{align}\label{eq:cdf:gam2:def}
F_{\gamma_{2,\Ss}}(x) &= \Prob\left(\min\left(\gamma_{12,\Ss},\GMRS,\GRFus\right)<x\right)\nonumber\\
&=1 - \Prob\left(\gamma_{12,\Ss}>x\right)\Prob\left(\GMRS>x\right)\Prob\left(\GRFus>x\right).
\end{align}

Based on the AS metric in~\eqref{eq:ASc near}, for $x<\frac{a_2}{a_1}$, we have
\begin{align}\label{eq:cdf:gam12S1}
\Prob\left(\gamma_{12,\Ss}\!>\!x\right)
&=\frac{\MB\MT}{\bGRNu}\sum_{p=0}^{\MB-1}\frac{(-1)^p\binom{\MB-1}{p}e^{-\frac{(p+1)\cnox}
{\bGSNu}}}{p+1}\int_{0}^{\infty} e^{-\left(\frac{(p+1)\cnox }{\bGSNu}+\frac{\MT}{\bGRNu}\right)y} dy\nonumber\\
&\hspace{0em}=\MB\sum_{p=0}^{\MB-1}\frac{(-1)^p\binom{\MB-1}{p}e^{-\frac{(p+1)\cnox}{\bGSNu}}}
{(p+1)\left(1+\frac{\bGRNu}{\bGSNu}\frac{(p+1)\cnox}{\MT}\right)},
\end{align}
where the third equality follows since $f_{\GRNu}(y)=\frac{\MT}{\bGRNu}e^{-\frac{\MT}{\bGRNu}y}$ and $F_{\GSNu}(x) = (1-e^{-\frac{x}{\bGSNu}})^{\MB}$ and $F_{\GSNu}(x)$ can be expressed as
\begin{align}\label{eq:cdf:GSNu}
F_{\GSNu}(x) = 1-\MB\sum_{p=0}^{\MB-1}\frac{(-1)^p\binom{\MB-1}{p}}{p+1} e^{-\frac{(p+1)x}{\bGSNu}}.
\end{align}

Moreover, $\Prob\left(\GMRS>x\right)$ and $\Prob\left(\GRFus>x\right)$ are evaluated in~\eqref{eq:cdf:GMRS} and~\eqref{eq:cdf:GRFus}, respectively. Accordingly, $F_{\gamma_{2,\Ss}}(x)$ can be written as
\begin{align}\label{eq:cdf:gam2:S1}
F_{\gamma_{2,\Ss}}(x) &\!= 1 \!- \!\MR\MB e^{-\frac{x}{\bGRFu}}
\!\sum_{p=0}^{\MB-1}\!
\frac{(-1)^p\binom{\MB-1}{p}e^{-\frac{(p+1)\cnox}{\bGSNu}}}
{(p+1)\left(1\!+\frac{\bGRNu}{\bGSNu}\frac\cnox{\MT}\right)}
\!\sum_{q=0}^{\MR-1}\!
\frac{(-1)^q\binom{\MR-1}{q}e^{-\frac{(q+1)\cnox}{\bGBR}}}
{(q+1)\left(1\!+\frac{\bGSI}{\bGBR}(q+1)\cnox\right)}.
\end{align}

Now,  the ergodic rate of $\SUuu$ can be  obtained by substituting the derived $F_{\gamma_{2,\Ss}}$ in~\eqref{eq:cdf:gam2:S1} into~\eqref{eq:rate and cdf} and performing some algebraic manipulations. For $x>\frac{a_2}{a_1}$ it can be readily checked that $F_{\gamma_{2,\Ss}}(x)= 1$ and hence $\mathcal{R}_{\SUuu}^{\Ss}=0$.

%
%
%


\vspace{-0.4em}
\bibliographystyle{IEEEtran}

\end{document}